\documentclass[amsmath,twocolumn,prb,groupedaddress,longbibliography,notitlepage]{revtex4-2}
\usepackage[colorlinks=true,allcolors=blue]{hyperref}
\usepackage{amsmath}
\usepackage{amssymb}
\usepackage{amsthm}
\usepackage{xcolor}
\usepackage{graphicx}
\usepackage{float}
\usepackage{subfigure}
\usepackage{hhline}
\newtheorem{theorem}{Proposition}
\newtheorem{lemma}{Lemma}
\newtheorem{corollary}{Corollary}

\begin{document}

\title{Lee--Yang zeros in the Rydberg atoms}

\author{Chengshu Li}
\email{lichengshu272@gmail.com}
\affiliation{Institute for Advanced Study, Tsinghua University, Beijing 100084, China}

\author{Fan Yang}
\affiliation{Institute for Advanced Study, Tsinghua University, Beijing 100084, China}

\date{\today}

\begin{abstract}
Lee--Yang (LY) zeros play a fundamental role in the formulation of statistical physics in terms of (grand) partition functions, and assume theoretical significance for the phenomenon of phase transitions. In this paper, motivated by recent progress in cold Rydberg atom experiments, we explore the LY zeros in classical Rydberg blockade models. We find that the distribution of zeros of partition functions for these models in one dimension (1d) can be obtained analytically. We prove that all the LY zeros are real and negative for such models with arbitrary blockade radii. Therefore, no phase transitions happen in 1d classical Rydberg chains. We investigate how the zeros redistribute as one interpolates between different blockade radii. We also discuss possible experimental measurements of these zeros.\\

\end{abstract}
\maketitle

\section{Introduction}
Much of modern statistical physics builds on the concept of (grand) partition functions. For finite systems in general, they are positive, analytic functions in terms of physical quantities, and, in particular, are smooth functions with regard to temperature $T$ when the latter is finite. Historically, this fact caused some confusion as to how discontinuity, as occurs in phase transitions, can emerge from such a smooth function, before it became clear that phase transitions can only happen in the thermodynamic limit. However, it was not until the work of Lee and Yang~\cite{Yang1952,Lee1952} when a more detailed description of the relation between the analyticity of partition functions and phase transitions was put forward. 

Exactly 70 years ago in 1952, in two seminal papers~\cite{Yang1952,Lee1952}, Lee and Yang proposed to study the complex zeros of the (grand) partition function as a polynomial of some sort of fugacity, and showed how the distribution of such zeros in the thermodynamic limit relates to the existence and properties of phase transitions. They proved that if there exists a region containing part of the positive real axis with no zeros inside, in the corresponding regime of physical parameters no phase transition can happen. Conversely, if the zeros accumulate to a point on the positive real axis in the thermodynamic limit, phase transitions will generally ensue. As a concrete example, they considered the Ising model with arbitrary ferromagnetic couplings, and proved that all the zeros are located on the unit circle with the fugacity defined as $\exp(-\beta h)$, where $h$ is the magnetic field. A direct corollary is that phase transitions are not possible unless $h=0$ in the Ising model. 

The original construction of Lee and Yang focused on spin-1/2 Ising ferromagnets, or equivalently attractive lattice gases. Generalizations of this program to higher spins and Heisenberg models soon followed~\cite{Asano1968, Suzuki1968a, Suzuki1968b, Griffiths1969, Asano1970, Ruelle1971, Suzuki1971, Kurtze1978}. In the next decades, results on other geometries and interactions, in particular antiferromagnetic Ising models, became available~\cite{Lieb1972,Heilmann1972,Dobrushin1985,Beauzamy1996,Kim2004,Hwang2010,Lebowitz2012,Lebowitz2016}. There are also more recent theoretical progresses as well as generalizations to dynamical and/or quantum systems~\cite{Heyl2013,Brandner2017,Deger2019,Deger2020,Kist2021}. The endeavor of characterizing complex zeros of real polynomials with positive coefficients has attracted interest from the mathematics community alike~\cite{Kurtz1992,Borcea2009a,Borcea2009b,Ruelle2010}.

While originally a purely theoretical discussion with no direct experimental relevance envisioned, Lee--Yang (LY) zeros turn out experimentally measurable using nuclear magnetic resonance (NMR) techniques~\cite{Wei2012,Peng2015}. The idea is to carefully couple a collection of spins to a probe spin, and extract the zeros from its quantum evolution. Recently, cold Rydberg atoms have attracted extensive research efforts for their versatile control and measurements together with strong dipole interactions. Two major frontiers are unfolding --- First, various novel phases and phase transitions emerge from geometries enabled by single-atom manipulations, including one dimensional (1d) chains~
\cite{Bernien2017,Keesling2019} and two dimensional (2d) square~\cite{Satzinger2021,Ebadi2021, Samajdar2020, Kalinowski2022} and kagome lattices~\cite{Verresen2021,Semeghini2021,Samajdar2021,Cheng2021,Giudici2022}. Even in the seemingly innocuous 1d case, there are some debates on the nature of the phase transitions~\cite{Fendley2004, Samajdar2018, Giudici2019, Chepiga2019,Rader2019,Maceira2022}, a caricature of the rich and subtle physics in such systems. Second, the constrained Hilbert space drastically changes the dynamical behavior, leading to violations of the eigenstate thermolization hypothesis (ETH)~\cite{Bernien2017}, which is in turn closely related to the quantum many body scar states~\cite{Turner2018,Serbyn2021}.

Motivated by these exciting progresses, we study how the LY zeros distribute in classical Rydberg atom systems. Here by ``classical'' we mean that all the terms in the Hamiltonian commute with each other, yet the Hilbert space is still constrained by the Rydberg blockade. One of the important features of these models, as we discover in our analysis, is that the distribution of their LY zeros can be obtained rigorously. This is in contrast to general cases, where the distribution of LY zeros does not have a clear pattern. We provide a rigorous proof that in 1d for any blockade radius the zeros are real and negative.  Therefore, no phase transitions occur at finite temperature. This is in agreement with the general belief that classical phase transitions do not happen in 1d systems at finite temperature.
Our analysis adds the classical 1d Rydberg chain to a small number of models whose distribution of zeros of the partition function can be obtained  analytically, making the study of the mathematical structure of such models interesting in its own right.
Beyond rigorous analytical discussions, we also numerically investigate how the zeros redistribute as the Rydberg blockade interpolates between different radii. Finally, inspired by previous experimental efforts, we propose a way to measure the LY zeros of the Rydberg chain, a seemingly abstract quantity, in cold atom experiments, bringing our predictions to physical reality.

\section{Rydberg blockade Hamiltonian, partition function, and LY zeros}
We begin with the classical Rydberg blockade Hamiltonian,
\begin{equation}\label{eq:HPXP}
H=-\Delta\sum_i n_i,
\end{equation}
where $n_i=0,1$ is the number of Rydberg excitation on each site $i$ and $\Delta$ is the detuning. Generally there is a Rabi oscillation term that couples the ground state and the Rydberg state, which we assume to be small and ignore here. The effect of the Rydberg blockade is taken into account by requiring that no two atoms can be excited within a radius of $r$, or that there can be at most one Rydberg atom in each consecutive $(r+1)$ sites. For a chain with $n$ sites and open boundary conditions, the number of different configurations with $m$ Rydberg atoms, $F_n^m$, is
\begin{equation}
F_n^m=
\begin{pmatrix}
n-r(m-1)\\
m
\end{pmatrix}
\end{equation}
from elementary combinatorics. These coefficients satisfy a recursion relation
\begin{equation}
F_{n+1}^m=F_{n}^m+F_{n-r}^{m-1}.
\end{equation}
To see this, notice that either the $(n+1)$-th site is unoccupied, in which case one is free to put $m$ Rydberg atoms in the first $n$ sites, giving the first term, or the $(n+1)$-th site is occupied, blocking sites $n-r+1$ to $n$, giving the second term. To proceed we introduce the fugacity $y=\exp(\beta\Delta)$, and the partition function is
\begin{equation}
Z_n(y)=\sum_{m=0}^{\lfloor(n+r)/(r+1)\rfloor}F_n^m y^m,
\end{equation}
where $\lfloor\cdot\rfloor$ denotes the floor function. The partition function in turn satisfies
\begin{equation}
Z_{n+1}(y)=Z_{n}(y)+yZ_{n-r}(y),\label{eq:rec}
\end{equation}
together with initial conditions
\begin{equation}
Z_n(y)=1+ny,\ 1\leq n\leq r+1.
\end{equation}
From now on we deem $y$ as a complex variable and explore the complex roots (zeros) of the partition function as a polynomial in $y$. The first main result we obtain is 
\begin{theorem}
For any $r,n$, all the zeros of the partition function as a polynomial of $y$ are real and negative.\label{prop:1}
\end{theorem}
The proof is elementary yet slightly involved. Readers who are not interested in the proof can skip to the next section.

\begin{figure}
\centering
\includegraphics[width=0.9\columnwidth]{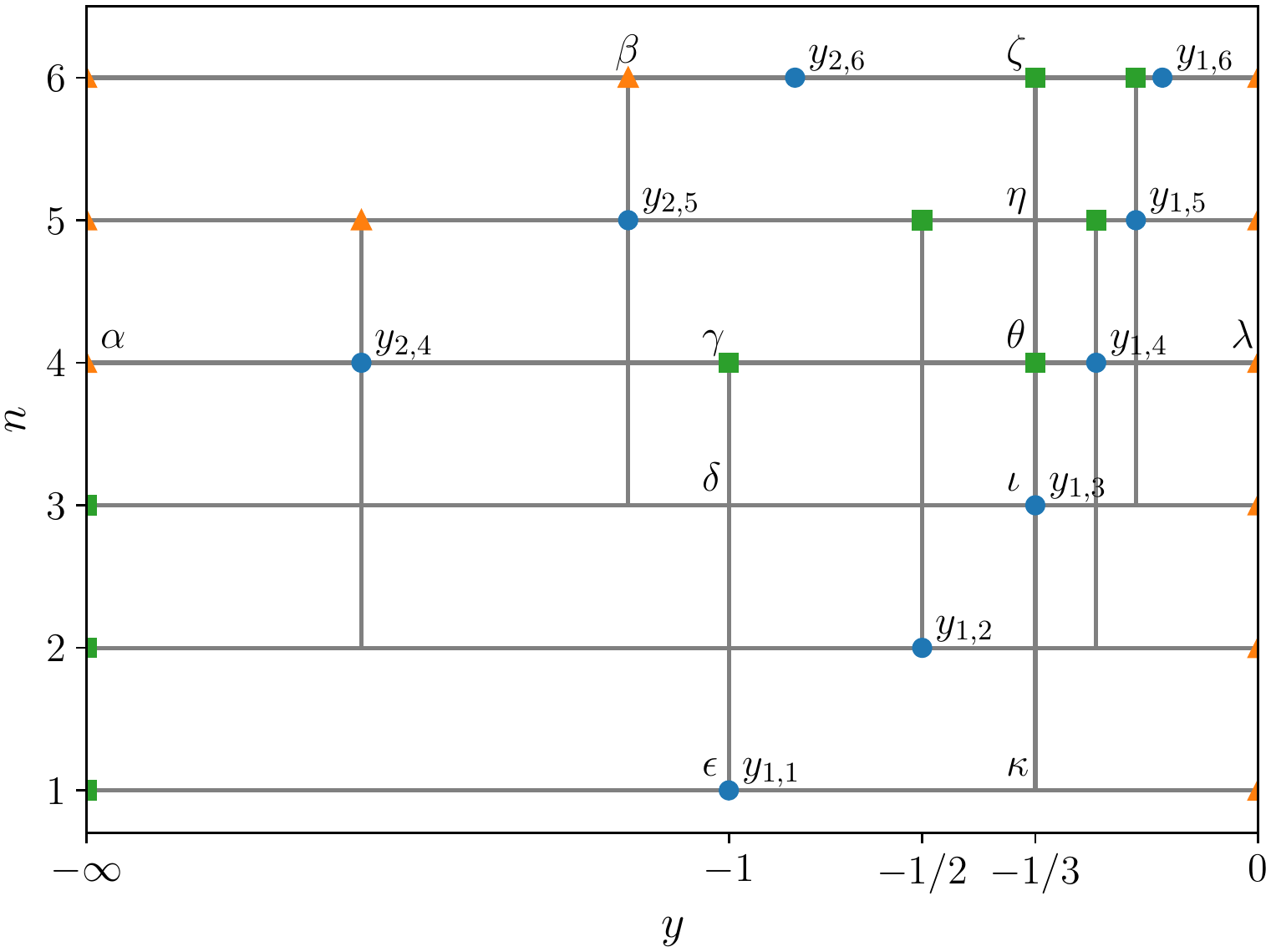}
\caption{An illustration of the proof of Lemma \ref{lemma} for the case $r=2$. Blue dots denote the zeros and orange triangles (green squares) indicate that the polynomial is positive (negative) at that point. As a result of the recursion relation (\ref{eq:rec}), the zeros in each group of three polynomials (e.g.,~$y_{1,1}, y_{1,2},y_{1,3}$) determine the sign of the polynomials in the next group, which in turn guarantees the existence of zeros in between.
A few examples are in order. We label the values of the polynomials with Greek letters as shown in the figure. First, $\delta<0$ and $\epsilon=0$ give $\gamma<0$, which together with $\alpha>0$ fixes $y_{2,4}$. Then, $\iota=0$ and $\kappa>0$ give $\theta<0$, which together with $\lambda>0$ fixes $y_{1,4}$. As a final example, $\eta<0$ and $\iota=0$ give $\zeta<0$, which together with $\beta>0$ fixes $y_{2,6}$. Note that we only mark the zeros and the points which are \emph{directly} used to fix the zeros, e.g., $\alpha$, $\beta$, $\gamma$, $\epsilon$, $\zeta$, $\theta$, $\iota$, and $\lambda$.}
\label{fig:proof}
\end{figure}

The main idea of the proof is as follows, which is similar to the one given in~\cite{Heilmann1972}: for any fixed $r$, we already know that the partition function is a polynomial of degree $\lfloor(n+r)/(r+1)\rfloor$. That is, as $n$ increases, the first $(r+1)$ polynomials have degree one, the next $(r+1)$ ones degree two, and so on. The claim is proved by finding $\lfloor(n+r)/(r+1)\rfloor$ real roots of the polynomials in each group. This follows from the following lemma.
\begin{lemma}\label{lemma}
For each $i\geq0$ and $j=1,\cdots,r+1$, the $i+1$ roots of $Z_{n=(r+1)i+j}(y)$ are real and different. Furthermore, when they are ordered by $0>y_{1,n}>y_{2,n}>\cdots>y_{i+1,n}$, they satisfy
\begin{equation}
\begin{split}
&y_{1,(r+1)i+r+1}>y_{1,(r+1)i+r}>\cdots>y_{1,(r+1)i+1}\\
>{}&y_{2,(r+1)i+r+1}>y_{2,(r+1)i+r}>\cdots>y_{2,(r+1)i+1}\\
>{}&\cdots\\
>{}&y_{i+1,(r+1)i+r+1}>y_{i+1,(r+1)i+r}>\cdots>y_{i+1,(r+1)i+1}.
\end{split}
\end{equation}
\end{lemma}
\begin{proof}
We first prove the case $r=2$ for simplicity, see Fig.~\ref{fig:proof}, but the same proof works for any $r$, see below. The first 3 polynomials ($i=0$) obviously satisfy the lemma. If the lemma is satisfied for $Z_{3i+1}, Z_{3i+2}, Z_{3i+3}$, from continuity we know the sign of the polynomial in each segment partitioned by the zeros. For example, $Z_{3i+k}(y)>0$ when $y_{1,3i+k}<y<0$, $k=1,2,3$. We obtain from the recursion relation
\begin{equation}
Z_{3(i+1)+1}(0)>0
\end{equation}
and
\begin{equation}
\begin{split}
&Z_{3(i+1)+1}(y_{1,3i+3})\\
={}&Z_{3i+3}(y_{1,3i+3})+y_{1,3i+3} Z_{3i+1}(y_{1,3i+3})\\
={}&y_{1,3i+3} Z_{3i+1}(y_{1,3i+3})<0.
\end{split}
\end{equation}
Therefore, there exists a zero of $Z_{3(i+1)+1}$, namely $y_{1,3(i+1)+1}$, in the open interval $(y_{1,3i+3},0)$ by the intermediate value theorem. Similar calculations give the other $i+1$ zeros, $y_{m,3(i+1)+1}$, in the intervals $(y_{m,3i+3},y_{m-1,3i+1})$, where $y_{i+2,3i+3}=-\infty$. Next, the zeros of $Z_{3(i+1)+2},Z_{3(i+1)+3}$ can be found in the same way, giving inequalities including $y_{m,3(i+1)+1}<y_{m,3(i+1)+2}<y_{m-1,3i+2}$ and $y_{m,3(i+1)+2}<y_{m,3(i+1)+3}<y_{m-1,3i+3}$. Combining the results above, we find
\begin{equation}
\begin{split}
y_{m,3(i+1)+1}&{}<y_{m,3(i+1)+2}<y_{m,3(i+1)+3}\\
&{}<y_{m-1,3i+3}<y_{m-1,3(i+1)+1},
\end{split}
\end{equation}
which concludes the proof for $r=2$. 

The proof for $r\geq3$ follows exactly the same spirit and is presented below for the readers who wish to read a general proof. To facilitate comparison with the special case of $r=2$ above, we translate the proof to general $r$ verbatim as follows.

The first $(r+1)$ polynomials ($i=0$) obviously satisfy the lemma. If the lemma is satisfied for $Z_{i(r+1)+j},j=1,\dots,r+1$, from continuity we know the sign of the polynomial in each segment partitioned by the zeros. For example, $Z_{i(r+1)+k}(y)>0$ when $y_{1,i(r+1)+k}<y<0$, $k=1,\dots,r+1$. We obtain from the recursion relation
\begin{equation}
Z_{(i+1)(r+1)+1}(0)>0
\end{equation}
and
\begin{equation}
\begin{split}
&Z_{(i+1)(r+1)+1}(y_{1,i(r+1)+r+1})\\
={}&Z_{i(r+1)+r+1}(y_{1,i(r+1)+r+1})\\
&+y_{1,i(r+1)+r+1} Z_{i(r+1)+1}(y_{1,i(r+1)+r+1})\\
={}&y_{1,i(r+1)+r+1} Z_{i(r+1)+1}(y_{1,i(r+1)+r+1})<0.
\end{split}
\end{equation}
Therefore, there exists a zero of $Z_{(i+1)(r+1)+1}$, namely $y_{1,(i+1)(r+1)+1}$, in the open interval $(y_{1,i(r+1)+r+1},0)$ by the intermediate value theorem. Similar calculations give the other $i+1$ zeros, $y_{m,(i+1)(r+1)+1}$, in the intervals $(y_{m,i(r+1)+r+1},y_{m-1,i(r+1)+1})$, where $y_{i+2,i(r+1)+r+1}=-\infty$. Next, the zeros of $Z_{(i+1)(r+1)+j},j=2,\dots,r+1$ can be found in the same way, giving inequalities including $y_{m,(i+1)(r+1)+j}<y_{m,(i+1)(r+1)+j+1}<y_{m-1,i(r+1)+j+1},j=1,\dots,r$. Combining the results above, we find
\begin{equation}
\begin{split}
&y_{m,(i+1)(r+1)+1}<y_{m,(i+1)(r+1)+2}<\cdots\\
<{}&y_{m,(i+1)(r+1)+r+1}
<y_{m-1,i(r+1)+r+1}\\
<{}&y_{m-1,(i+1)(r+1)+1},
\end{split}
\end{equation}
which concludes the proof for general $r$.
\end{proof}

This lemma leads directly to Proposition \ref{prop:1}, as we have found all the roots of the partition function. As stated in the lemma, they are all real and negative.

We note that in two previous papers~\cite{Alcaraz2020a,Alcaraz2020b} on integrable systems generalized from free (para)fermions~\cite{Fendley2019}, a result identical to Proposition~\ref{prop:1} was conjectured based on numerical evidence, but a rigorous proof was lacking. The above proof fills this gap and finds applications in the field of integrable systems.

\section{Mathematical structure of the distribution of LY zeros}\label{sec:other}
It is possible to formally solve the recursion relation in Eq.~\eqref{eq:rec}. We first extend the initial condition to negative $n$ for convenience. One easily sees that
\begin{equation}
\begin{split}
&Z_n(y)=1,\ -r\leq n\leq 0,\\
&Z_n(y)=0,\ -2r\leq n\leq -r-1.
\end{split}
\end{equation}
We then have
\begin{equation}
\begin{split}
Z_n&=(1,0,\cdots,0)
\begin{pmatrix}
Z_n\\Z_{n-1}\\\vdots\\Z_{n-r}
\end{pmatrix}\\
&=(1,0,\cdots,0)
\begin{pmatrix}
1 & & & & y\\
1 & & & & \\
& 1 & & & \\
& & \ddots & & \\
& & & 1 & 
\end{pmatrix}^{n+r}
\begin{pmatrix}
Z_{-r}\\Z_{-r-1}\\\vdots\\Z_{-2r}
\end{pmatrix}\\
&=(1,0,\cdots,0)
\begin{pmatrix}
1 & & & & y\\
1 & & & & \\
& 1 & & & \\
& & \ddots & & \\
& & & 1 & 
\end{pmatrix}^{n+r}
\begin{pmatrix}
1\\0\\\vdots\\0
\end{pmatrix}.
\end{split}
\end{equation}
To proceed we need to diagonalize
\begin{equation}
A=\begin{pmatrix}
1 & & & & y\\
1 & & & & \\
& 1 & & & \\
& & \ddots & & \\
& & & 1 & 
\end{pmatrix}.
\end{equation}
The characteristic polynomial reads
\begin{equation}
\lambda^{r+1}-\lambda^{r}-y=0.\label{eq:ch_poly}
\end{equation}
Assuming all the eigenvalues $\lambda_i$ to be different (we will verify this soon), we see that $A$ is diagonalized by 
\begin{equation}
U=\begin{pmatrix}
1 & 1 & \cdots & 1\\
\lambda_1^{-1} & \lambda_2^{-1} & \cdots & \lambda_{r+1}^{-1}\\
\lambda_1^{-2} & \lambda_2^{-2} & \cdots & \lambda_{r+1}^{-2}\\
\vdots & \vdots & \ddots & \vdots\\
\lambda_1^{-r} & \lambda_2^{-r} & \cdots & \lambda_{r+1}^{-r}
\end{pmatrix},
\end{equation}
which is a Vandermonde matrix of $\lambda_i^{-1}$. We arrive at
\begin{equation}
\begin{split}
Z_n&=(U\,\mathrm{diag}\{\lambda_1^{n+r},\cdots,\lambda_{r+1}^{n+r}\}\,U^{-1})_{1,1}\\
&=\frac{1}{\prod_i\lambda_i}\sum_i\frac{\lambda_i^{n+r+1}}{\prod_j'(\lambda_j^{-1}-\lambda_i^{-1})}.
\end{split}
\end{equation}
For $r=1$, or nearest-neighbor blockade, $\lambda_{1,2}=(1\pm\sqrt{1+4y})/2$,
\begin{equation}
\begin{split}
Z_n&=\frac{\lambda_1^{n+2}-\lambda_2^{n+2}}{\lambda_1\lambda_2(\lambda_2^{-1}-\lambda_1^{-1})}
=\frac{\lambda_1^{n+2}-\lambda_2^{n+2}}{\lambda_1-\lambda_2}\\
&=\frac{1}{\sqrt{4y+1}}\\
&\times\bigg(\Big(\frac{1+\sqrt{4y+1}}{2}\Big)^{n+2}
-\Big(\frac{1-\sqrt{4y+1}}{2}\Big)^{n+2}\bigg).
\end{split}\label{eq:r1}
\end{equation}
To find the zeros we set $Z_n(y)=0$ and get
\begin{equation}
\Big(\frac{1+\sqrt{4y+1}}{1-\sqrt{4y+1}}\Big)^{n+2}=1.
\end{equation}
It's straightforward to find
\begin{equation}
y=-\frac{1+\tan^2(\phi/2)}{4},
\end{equation}
where
\begin{equation}
\phi=2\pi m/(n+2), m=1,\cdots,\lfloor(n+1)/2\rfloor,
\end{equation}
which is real, negative and always smaller than $-1/4$.

In order for Eq.~\eqref{eq:ch_poly} to have degenerate roots, the latter must also be roots of the derivative of the original polynomial. This can only happen when $\lambda=r/(r+1)$, and correspondingly $y_0=-r^r/(r+1)^{r+1}$. This $y$ is also where the number of real roots of Eq.~\eqref{eq:ch_poly} changes, hinting at its special role. Indeed, we have the following result.
\begin{theorem}\label{prop:2}
All the zeros of the partition function as in Proposition~\ref{prop:1} are smaller than $y_0$. 
\end{theorem}
\begin{proof}
We only need to prove that when $y_0\leq y<0$, $Z_n(y)$ is always greater than zero. We prove it by showing a stronger condition, 
\begin{equation}
\frac{Z_{n+1}}{Z_{n}}>\frac{r}{r+1}.
\end{equation}
This obviously holds for $Z_n, n=-r,\cdots,-1$. Assume this to be true for $Z_n, n=i,\cdots,i+r-1$, then
\begin{equation}
\frac{Z_{i+r}}{Z_{i}}>\Big(\frac{r}{r+1}\Big)^r.
\end{equation}
The recursion relation gives
\begin{equation}
\frac{Z_{i+r+1}}{Z_{i+r}}=1+y\frac{Z_{i}}{Z_{i+r}}>1+y\Big(\frac{r+1}{r}\Big)^r.
\end{equation}
Then we have
\begin{equation}
\frac{Z_{i+r+1}}{Z_{i+r}}>\frac{r}{r+1}\Leftarrow
1+y\Big(\frac{r+1}{r}\Big)^r\geq\frac{r}{r+1} \Leftrightarrow y\geq y_0.
\end{equation}
\end{proof}
We also have the following corollary:
\begin{corollary}
For $y_0\leq y<0$, we have $Z_n(y)\to0$ as $n\to\infty$.
\end{corollary}
\begin{proof}
Using the recursion relation and Proposition \ref{prop:2}, we have
\begin{equation}\label{decrease}
    \frac{Z_{n+1}(y)}{Z_n(y)}=1+y\frac{Z_{n-r}(y)}{Z_n(y)}<1.
\end{equation}
Thus, $Z_n$ decreases monotonically with index $n$.
By realizing
\begin{equation}
    \frac{Z_{n-r}(y)}{Z_n(y)}>1,
\end{equation}
Eq. (\ref{decrease}) can be made stronger 
\begin{equation}
    \frac{Z_{n+1}(y)}{Z_n(y)}=1+y\frac{Z_{n-r}(y)}{Z_n(y)}<1+y.
\end{equation}
At last, the squeeze theorem gives
\begin{equation}
    \lim_{n\to\infty}Z_n(y)=0, \forall y\in[y_0,0).
\end{equation}
\end{proof}

\section{Interpolations between different blockade radii}
To better understand Rydberg chain systems, we investigate how the zeros redistribute as the blockade radius increases. To begin with we focus on the case $r=1$ to $r=2$, or nearest- to next-nearest-neighbor blockade. For this purpose we add a next-nearest-neighbor interacting term to the $r=1$ Hamiltonian and the new Hamiltonian reads
\begin{equation}\label{eq:interpolation}
H'=-\Delta\sum_i n_i+V\sum_i n_i n_{i+2},
\end{equation}
with the same condition that no two nearest neighbors can be excited simultaneously.
This is the classical ($w=0$) version of the Fendley--Sengupta--Sachdev model~\cite{Fendley2004}. We note that this model is different from the classical Rydberg blockade model discussed in previous sections. However, the same physics should be recovered in the following limits: $V=0$ corresponds to blockade radius $r=1$ and $V=\infty$ corresponds to $r=2$.
The partition function of the new Hamiltonian is then
\begin{equation}
Z'=\sum_n y^n \sum_{\{n_i\}}\exp(-\beta V\sum_i n_i n_{i+2}).\label{eq:r12}
\end{equation}

\begin{figure}
\centering
\includegraphics[width=0.95\columnwidth]{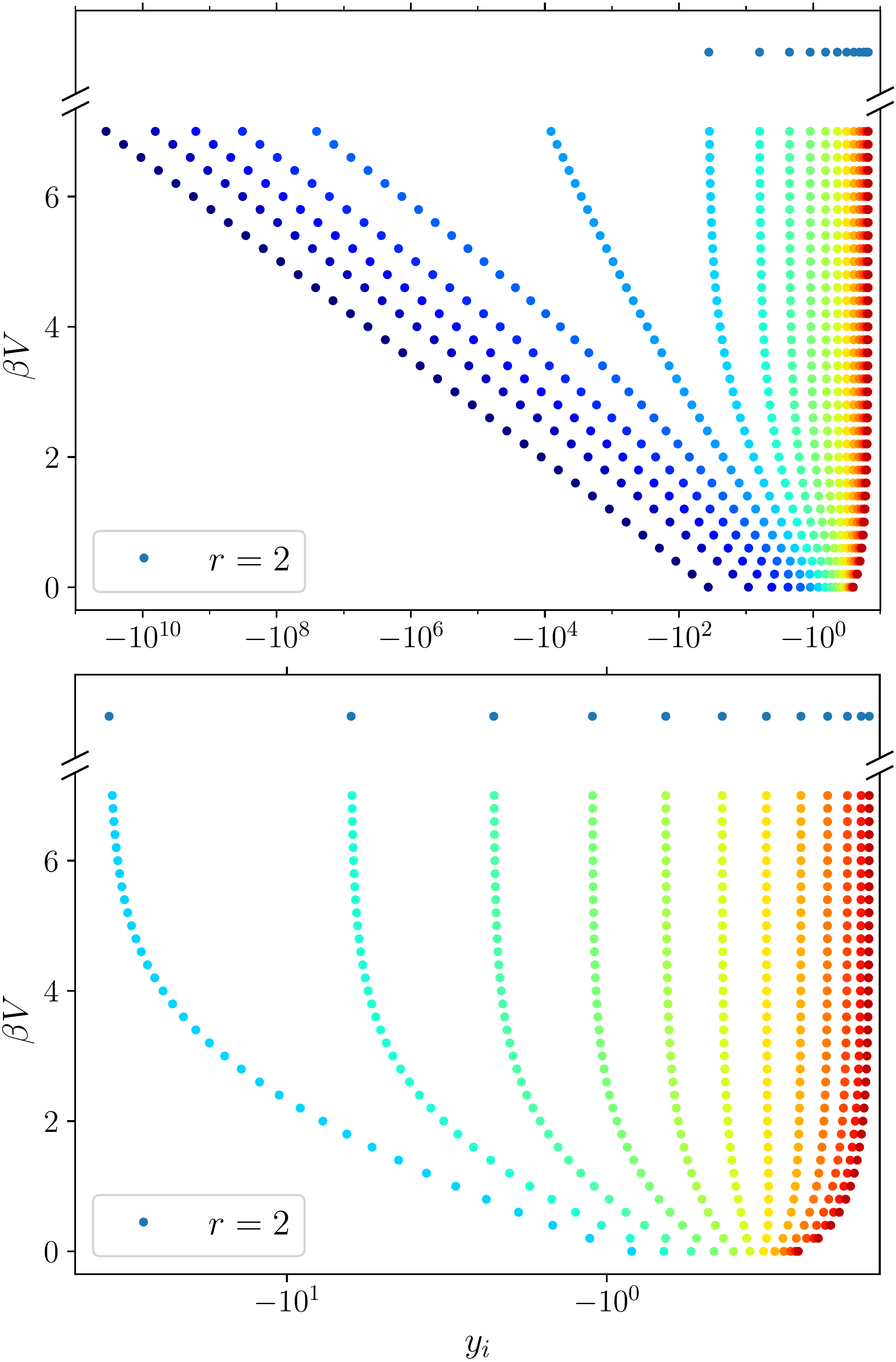}
\caption{The zeros of the partition function \eqref{eq:r12} for $n=36$. The largest $n/3=12$ of them converge to that of $r=2$ when $V$ goes to $\infty$, while the others diverge to $-\infty$ exponentially, which are omitted in the lower panel.}
\label{fig:numerics}
\end{figure}

\begin{figure}
\centering
\includegraphics[width=0.95\columnwidth]{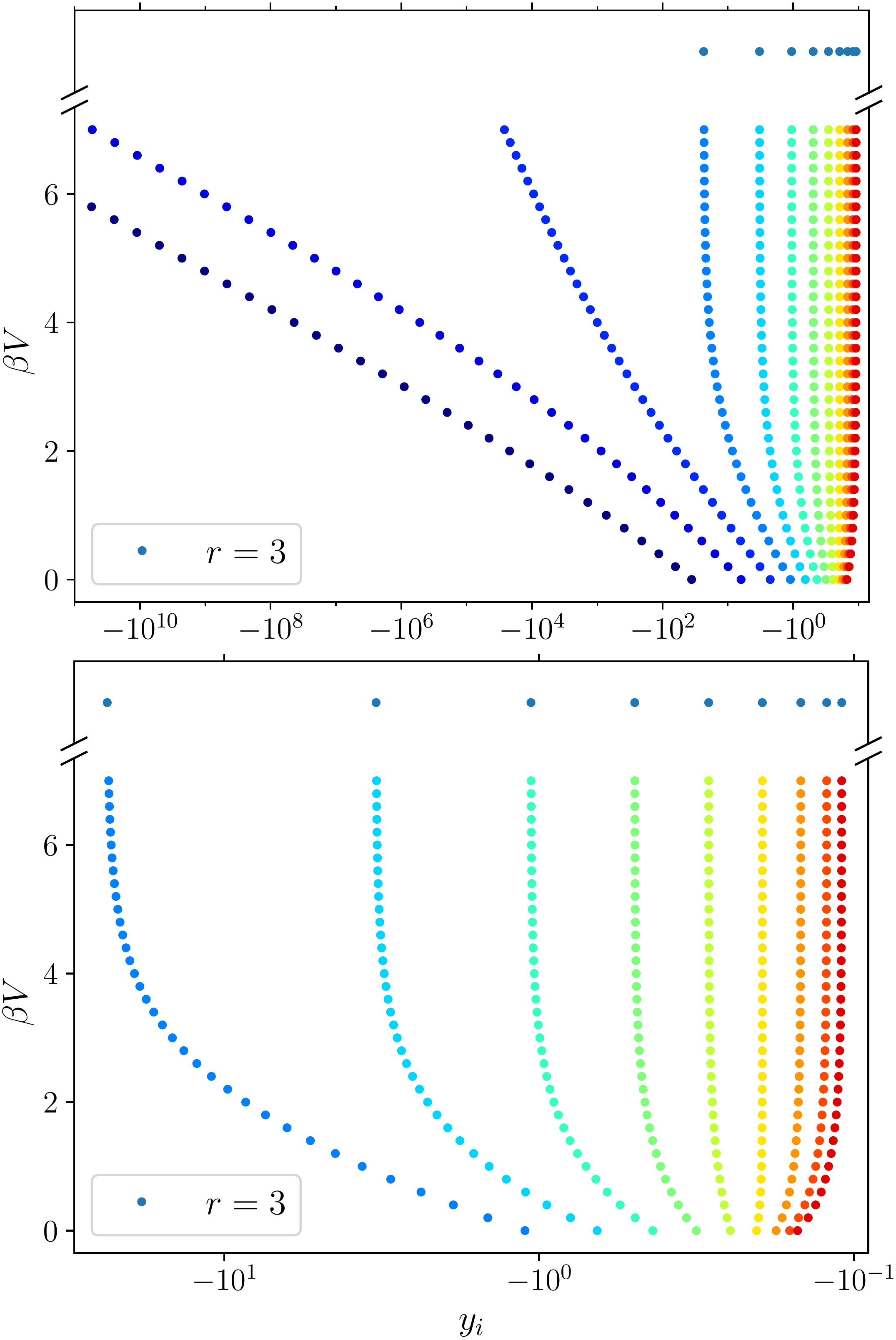}
\caption{The zeros of the partition function \eqref{eq:r23} for $n=36$. The largest $n/4=9$ of them converge to that of $r=3$ when $V$ goes to $\infty$, while the others diverge to $-\infty$ exponentially, which are omitted in the lower panel.}
\label{fig:numerics2}
\end{figure}

The Hamiltonian (\ref{eq:interpolation}) facilitates the understanding of Rydberg chains  in the following two aspects. First, it helps us understand Rydberg chains beyond the classical Rydberg blockade model. In this aspect, the distribution of LY zeros of partition function (\ref{eq:r12}) is obtained numerically. We calculate the zeros for $n=24,\cdots,36$. While for negative $V$ the zeros are in general complex, we find that when $V$ is positive the zeros are real. Therefore, we conjecture that the LY zeros remain living on the negative real axis for $V>0$. Consequently, no phase transitions happen at finite temperature for such systems with repulsive interactions.

Second, it sheds light on how the LY zeros of classical Rydberg blockade models with different blockade radii connect with each other.
As stated in the previous sections, for models with blockade radius $r=1$ and $r=2$, the number of LY zeros are $\lfloor (n+1)/2 \rfloor$ and $\lfloor (n+2)/3 \rfloor$, respectively.
By numerically studying how the LY zeros evolve during the interpolation given by Eq. (\ref{eq:interpolation}), we can gain some insight on how the LY zeros of Rydberg blockade models with different blockade radii are connected with each other, and how the number of LY zeros changes as $r$ changes to $r+1$.

Among the $\lfloor (n+1)/2 \rfloor$ LY zeros for finite $\beta V$, there is a sharp difference between the distributions of the largest $\lfloor (n+2)/3\rfloor$ zeros and the others as $V$ goes to $\infty$: the former converge to that of the $r=2$ result, while the latter diverge to $-\infty$ exponentially. The convergence of the larger roots is a ramification of the fact that as $\beta V$ goes to $\infty$ we should recover the physics of the model with blockade radius $r=2$. Thus, we conclude that as $r$ increases from 1 to 2, the largest $\lfloor (n+2)/3 \rfloor$ zeros of the blockade model with $r=1$ evolve into each of the LY zeros of the blockade model with $r=2$. The change in the number of LY zeros of classical Rydberg blockade models with different blockade radii is attributed to the fact that the rest of the zeros in the interpolation process eventually diverge to $-\infty$. Numerically, we find that the dichotomy becomes apparent with moderate $\beta V\simeq5$. The results for $n=36$ are shown in Fig.~\ref{fig:numerics}.

Due to the unified structures among different blockade radii discussed above, we expect the interpolations for general $r$ to behave similarly to that between $r=1$ and $r=2$. Indeed, we numerically check the one between $r=2$ and $r=3$ with an interpolation by a repulsive next-next-nearest-neighbor interaction
\begin{equation}
H''=-\Delta\sum_i n_i+V\sum_i n_i n_{i+3},
\end{equation}
with the condition that no two atoms within a radius of two can be excited simultaneously. 
The corresponding partition function reads
\begin{equation}
Z''=\sum_n y^n \sum_{\{n_i\}}\exp(-\beta V\sum_i n_i n_{i+3}),\label{eq:r23}
\end{equation}
where $n_i$ now sum over configurations with the $r=2$ constraint. The largest $\lfloor (n+3)/4\rfloor$ zeros approach the $r=3$ results, and the others diverge to $-\infty$ as $V$ goes to $\infty$, see Fig.~\ref{fig:numerics2}.

\section{Experimental measurement of the LY zeros}
Amenability to individual manipulations and measurements empowers cold Rydberg atoms as an ideal platform to measure LY zeros experimentally. Here we propose a setup similar to those in \cite{Wei2012,Peng2015} where the zeros can be measured dynamically. 

\begin{figure}
    \centering
    \includegraphics[width=0.9\columnwidth]{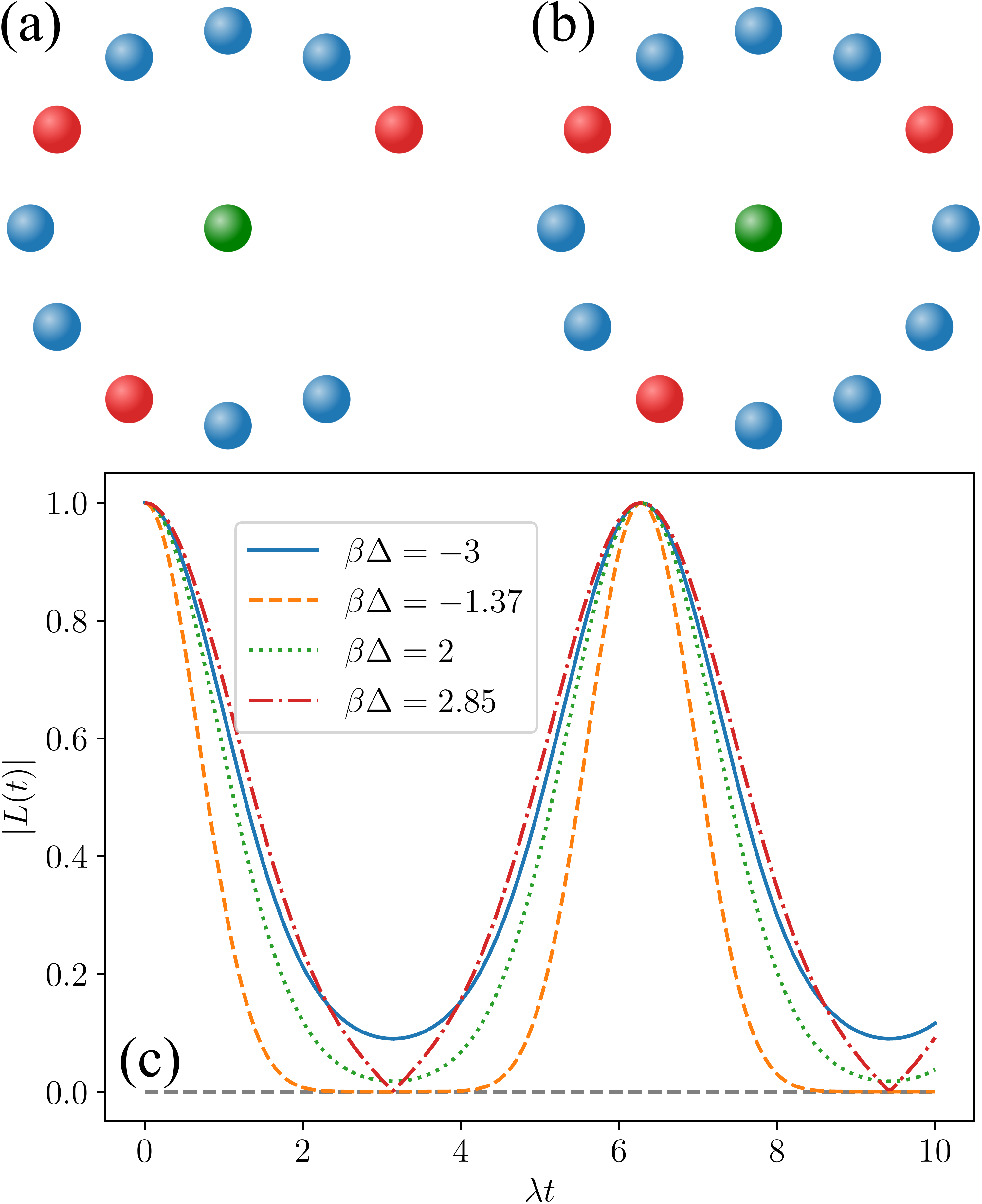}
    \caption{(a, b) The experimental setup for (a) open boundary condition and (b) periodic boundary condition. Blue (red) balls denote atoms in the ground (Rydberg) state, and the green ball denotes the probe atom. Typical time dependence of $|L(t)|$ as defined in Eq.~\eqref{eq:L} is shown in (c). Here we choose $n=24$, $r=1$. The difference between the zeros ($\beta\Delta=-1.37,2.85$) and non-zeros ($\beta\Delta=-3,2$) is clearly visible.}
    \label{fig:exp}
\end{figure}

We first assemble atoms (system, denoted by s) uniformly on a circle, either to a C-shaped arc for open boundary condition or an O-shaped circle for periodic boundary condition, with detuning $\Delta$ and inverse temperature $\beta$, see Fig.~\ref{fig:exp}(a, b). Then we put a probe (denoted by p) atom on the axis normal to the circle plane, not necessarily in the same plane as the circle to allow for more general couplings while maintaining a homogeneous one-to-all interaction,
\begin{equation}
H_\mathrm{sp}=\lambda n_\mathrm{p}\sum_i n_i=\lambda n_\mathrm{p}n_\mathrm{s},
\end{equation}
where $\lambda$ is a coupling constant and $n_\mathrm{p}$ denotes Rydberg number operator of the probe atom. By preparing the probe atom as a superposition of the ground state $|0\rangle$ and the Rydberg state $|1\rangle$, $(|0\rangle+|1\rangle)/\sqrt{2}$, it evolves with a phase factor proportional to the total Rydberg atom number, weighted by the density matrix of the system, giving rise to measurable quantities
\begin{equation}
\begin{split}
\langle\sigma^x_\mathrm{p}(t)\rangle
&=Z(y)^{-1}\sum_{n_\mathrm{s}} y^{n_\mathrm{s}}\sum_{\{n_i\}}e^{-\beta H_\mathrm{s}}\cos(\lambda n_\mathrm{s} t),\\
\langle\sigma^y_\mathrm{p}(t)\rangle
&=Z(y)^{-1}\sum_{n_\mathrm{s}} y^{n_\mathrm{s}}\sum_{\{n_i\}}e^{-\beta H_\mathrm{s}}\sin(\lambda n_\mathrm{s} t).
\end{split}
\end{equation}
Here, $\sigma^x_\mathrm{p}=|1\rangle\langle0|+|0\rangle\langle1|$ and $\sigma^y_\mathrm{p} =i(|0\rangle\langle1| -|1\rangle\langle0|)$ can be measured by first performing proper rotations to change to the $x,y$ basis. Combining the two results gives
\begin{equation}
L(t)=\langle\sigma^x_\mathrm{p}(t)+i\sigma^y_\mathrm{p}(t)\rangle=Z(e^{\beta\Delta+i\lambda t})/Z(e^{\beta\Delta}).\label{eq:L}
\end{equation}
By tuning $\Delta$ and $t$, one can identify when $L(t)$ goes to zero, which in turn gives the zeros of the partition function. In particular, the results we show in the previous sections indicate that $\lambda t=(2n+1)\pi, n\in\mathbb{N}$, see Fig.~\ref{fig:exp}(c).

\section{Conclusions}
Motivated by the recent advances in cold Rydberg atom experiments, we study LY zeros in one-dimensional classical Rydberg blockade systems. We find that the distribution of LY zeros of such models can be obtained analytically. We rigorously prove that for general blockade radii the LY zeros are real and negative. Thus, there are no phase transitions at finite temperature. As we have shown in this paper, the 1d Rydberg chain is one of a small number of known models, whose distribution of LY zeros can be obtained analytically. Thus, its mathematical structures may merit further examination. In addition, we numerically find that these zeros keep living in the negative real axis, when one turns on a next-nearest-neighbor repulsive interaction to interpolate between the nearest- and next-nearest-neighbor blockades. We propose a way to experimentally observe these zeros by coupling the system to a probe atom and performing dynamical measurements thereon.

\section*{Acknowledgments}
We thank Hui Zhai for helpful discussions. 
Both authors are supported by the International Postdoctoral Exchange Fellowship Program and the Shuimu Tsinghua Scholar Program.

\bibliography{biblio}

\begin{thebibliography}{51}%
\makeatletter
\providecommand \@ifxundefined [1]{%
 \@ifx{#1\undefined}
}%
\providecommand \@ifnum [1]{%
 \ifnum #1\expandafter \@firstoftwo
 \else \expandafter \@secondoftwo
 \fi
}%
\providecommand \@ifx [1]{%
 \ifx #1\expandafter \@firstoftwo
 \else \expandafter \@secondoftwo
 \fi
}%
\providecommand \natexlab [1]{#1}%
\providecommand \enquote  [1]{``#1''}%
\providecommand \bibnamefont  [1]{#1}%
\providecommand \bibfnamefont [1]{#1}%
\providecommand \citenamefont [1]{#1}%
\providecommand \href@noop [0]{\@secondoftwo}%
\providecommand \href [0]{\begingroup \@sanitize@url \@href}%
\providecommand \@href[1]{\@@startlink{#1}\@@href}%
\providecommand \@@href[1]{\endgroup#1\@@endlink}%
\providecommand \@sanitize@url [0]{\catcode `\\12\catcode `\$12\catcode
  `\&12\catcode `\#12\catcode `\^12\catcode `\_12\catcode `\%12\relax}%
\providecommand \@@startlink[1]{}%
\providecommand \@@endlink[0]{}%
\providecommand \url  [0]{\begingroup\@sanitize@url \@url }%
\providecommand \@url [1]{\endgroup\@href {#1}{\urlprefix }}%
\providecommand \urlprefix  [0]{URL }%
\providecommand \Eprint [0]{\href }%
\providecommand \doibase [0]{https://doi.org/}%
\providecommand \selectlanguage [0]{\@gobble}%
\providecommand \bibinfo  [0]{\@secondoftwo}%
\providecommand \bibfield  [0]{\@secondoftwo}%
\providecommand \translation [1]{[#1]}%
\providecommand \BibitemOpen [0]{}%
\providecommand \bibitemStop [0]{}%
\providecommand \bibitemNoStop [0]{.\EOS\space}%
\providecommand \EOS [0]{\spacefactor3000\relax}%
\providecommand \BibitemShut  [1]{\csname bibitem#1\endcsname}%
\let\auto@bib@innerbib\@empty
\bibitem [{\citenamefont {Yang}\ and\ \citenamefont {Lee}(1952)}]{Yang1952}%
  \BibitemOpen
  \bibfield  {author} {\bibinfo {author} {\bibfnamefont {C.~N.}\ \bibnamefont
  {Yang}}\ and\ \bibinfo {author} {\bibfnamefont {T.~D.}\ \bibnamefont {Lee}},\
  }\bibfield  {title} {\bibinfo {title} {Statistical theory of equations of
  state and phase transitions. {I}. {T}heory of condensation},\ }\href
  {https://doi.org/10.1103/PhysRev.87.404} {\bibfield  {journal} {\bibinfo
  {journal} {Phys. Rev.}\ }\textbf {\bibinfo {volume} {87}},\ \bibinfo {pages}
  {404} (\bibinfo {year} {1952})}\BibitemShut {NoStop}%
\bibitem [{\citenamefont {Lee}\ and\ \citenamefont {Yang}(1952)}]{Lee1952}%
  \BibitemOpen
  \bibfield  {author} {\bibinfo {author} {\bibfnamefont {T.~D.}\ \bibnamefont
  {Lee}}\ and\ \bibinfo {author} {\bibfnamefont {C.~N.}\ \bibnamefont {Yang}},\
  }\bibfield  {title} {\bibinfo {title} {Statistical theory of equations of
  state and phase transitions. {II}. {L}attice gas and {I}sing model},\ }\href
  {https://doi.org/10.1103/PhysRev.87.410} {\bibfield  {journal} {\bibinfo
  {journal} {Phys. Rev.}\ }\textbf {\bibinfo {volume} {87}},\ \bibinfo {pages}
  {410} (\bibinfo {year} {1952})}\BibitemShut {NoStop}%
\bibitem [{\citenamefont {Asano}(1968)}]{Asano1968}%
  \BibitemOpen
  \bibfield  {author} {\bibinfo {author} {\bibfnamefont {T.}~\bibnamefont
  {Asano}},\ }\bibfield  {title} {\bibinfo {title} {Generalization of the
  {L}ee-{Y}ang theorem},\ }\href {https://doi.org/10.1143/PTP.40.1328}
  {\bibfield  {journal} {\bibinfo  {journal} {Prog. Theor. Phys.}\ }\textbf
  {\bibinfo {volume} {40}},\ \bibinfo {pages} {1328} (\bibinfo {year}
  {1968})}\BibitemShut {NoStop}%
\bibitem [{\citenamefont {Suzuki}(1968{\natexlab{a}})}]{Suzuki1968a}%
  \BibitemOpen
  \bibfield  {author} {\bibinfo {author} {\bibfnamefont {M.}~\bibnamefont
  {Suzuki}},\ }\bibfield  {title} {\bibinfo {title} {Theorems on the {I}sing
  model with general spin and phase transition},\ }\href
  {https://doi.org/10.1063/1.1664546} {\bibfield  {journal} {\bibinfo
  {journal} {J. Math. Phys.}\ }\textbf {\bibinfo {volume} {9}},\ \bibinfo
  {pages} {2064} (\bibinfo {year} {1968}{\natexlab{a}})}\BibitemShut {NoStop}%
\bibitem [{\citenamefont {Suzuki}(1968{\natexlab{b}})}]{Suzuki1968b}%
  \BibitemOpen
  \bibfield  {author} {\bibinfo {author} {\bibfnamefont {M.}~\bibnamefont
  {Suzuki}},\ }\bibfield  {title} {\bibinfo {title} {Theorems on extended
  {I}sing model with applications to dilute ferromagnetism},\ }\href
  {https://doi.org/10.1143/PTP.40.1246} {\bibfield  {journal} {\bibinfo
  {journal} {Prog. Theor. Phys.}\ }\textbf {\bibinfo {volume} {40}},\ \bibinfo
  {pages} {1246} (\bibinfo {year} {1968}{\natexlab{b}})}\BibitemShut {NoStop}%
\bibitem [{\citenamefont {Griffiths}(1969)}]{Griffiths1969}%
  \BibitemOpen
  \bibfield  {author} {\bibinfo {author} {\bibfnamefont {R.~B.}\ \bibnamefont
  {Griffiths}},\ }\bibfield  {title} {\bibinfo {title} {Rigorous results for
  {I}sing ferromagnets of arbitrary spin},\ }\href
  {https://doi.org/10.1063/1.1665005} {\bibfield  {journal} {\bibinfo
  {journal} {J. Math. Phys.}\ }\textbf {\bibinfo {volume} {10}},\ \bibinfo
  {pages} {1559} (\bibinfo {year} {1969})}\BibitemShut {NoStop}%
\bibitem [{\citenamefont {Asano}(1970)}]{Asano1970}%
  \BibitemOpen
  \bibfield  {author} {\bibinfo {author} {\bibfnamefont {T.}~\bibnamefont
  {Asano}},\ }\bibfield  {title} {\bibinfo {title} {Theorems on the partition
  functions of the {H}eisenberg ferromagnets},\ }\href
  {https://doi.org/10.1143/JPSJ.29.350} {\bibfield  {journal} {\bibinfo
  {journal} {J. Phys. Soc. Jpn.}\ }\textbf {\bibinfo {volume} {29}},\ \bibinfo
  {pages} {350} (\bibinfo {year} {1970})}\BibitemShut {NoStop}%
\bibitem [{\citenamefont {Ruelle}(1971)}]{Ruelle1971}%
  \BibitemOpen
  \bibfield  {author} {\bibinfo {author} {\bibfnamefont {D.}~\bibnamefont
  {Ruelle}},\ }\bibfield  {title} {\bibinfo {title} {Extension of the
  {L}ee-{Y}ang circle theorem},\ }\href
  {https://doi.org/10.1103/PhysRevLett.26.303} {\bibfield  {journal} {\bibinfo
  {journal} {Phys. Rev. Lett.}\ }\textbf {\bibinfo {volume} {26}},\ \bibinfo
  {pages} {303} (\bibinfo {year} {1971})}\BibitemShut {NoStop}%
\bibitem [{\citenamefont {Suzuki}\ and\ \citenamefont
  {Fisher}(1971)}]{Suzuki1971}%
  \BibitemOpen
  \bibfield  {author} {\bibinfo {author} {\bibfnamefont {M.}~\bibnamefont
  {Suzuki}}\ and\ \bibinfo {author} {\bibfnamefont {M.~E.}\ \bibnamefont
  {Fisher}},\ }\bibfield  {title} {\bibinfo {title} {Zeros of the partition
  function for the {H}eisenberg, ferroelectric, and general {I}sing models},\
  }\href {https://doi.org/10.1063/1.1665583} {\bibfield  {journal} {\bibinfo
  {journal} {J. Math. Phys.}\ }\textbf {\bibinfo {volume} {12}},\ \bibinfo
  {pages} {235} (\bibinfo {year} {1971})}\BibitemShut {NoStop}%
\bibitem [{\citenamefont {Kurtze}\ and\ \citenamefont
  {Fisher}(1978)}]{Kurtze1978}%
  \BibitemOpen
  \bibfield  {author} {\bibinfo {author} {\bibfnamefont {D.~A.}\ \bibnamefont
  {Kurtze}}\ and\ \bibinfo {author} {\bibfnamefont {M.~E.}\ \bibnamefont
  {Fisher}},\ }\bibfield  {title} {\bibinfo {title} {The {Y}ang-{L}ee edge
  singularity in spherical models},\ }\href
  {https://doi.org/10.1007/BF01011723} {\bibfield  {journal} {\bibinfo
  {journal} {J. Stat. Phys.}\ }\textbf {\bibinfo {volume} {19}},\ \bibinfo
  {pages} {205} (\bibinfo {year} {1978})}\BibitemShut {NoStop}%
\bibitem [{\citenamefont {Lieb}\ and\ \citenamefont {Ruelle}(1972)}]{Lieb1972}%
  \BibitemOpen
  \bibfield  {author} {\bibinfo {author} {\bibfnamefont {E.~H.}\ \bibnamefont
  {Lieb}}\ and\ \bibinfo {author} {\bibfnamefont {D.}~\bibnamefont {Ruelle}},\
  }\bibfield  {title} {\bibinfo {title} {A property of zeros of the partition
  function for {I}sing spin systems},\ }\href
  {https://doi.org/10.1063/1.1666051} {\bibfield  {journal} {\bibinfo
  {journal} {J. Math. Phys.}\ }\textbf {\bibinfo {volume} {13}},\ \bibinfo
  {pages} {781} (\bibinfo {year} {1972})}\BibitemShut {NoStop}%
\bibitem [{\citenamefont {Heilmann}\ and\ \citenamefont
  {Lieb}(1972)}]{Heilmann1972}%
  \BibitemOpen
  \bibfield  {author} {\bibinfo {author} {\bibfnamefont {O.~J.}\ \bibnamefont
  {Heilmann}}\ and\ \bibinfo {author} {\bibfnamefont {E.~H.}\ \bibnamefont
  {Lieb}},\ }\bibfield  {title} {\bibinfo {title} {Theory of monomer-dimer
  systems},\ }\href {https://doi.org/10.1007/BF01877590} {\bibfield  {journal}
  {\bibinfo  {journal} {Commun. Math. Phys.}\ }\textbf {\bibinfo {volume}
  {25}},\ \bibinfo {pages} {190} (\bibinfo {year} {1972})}\BibitemShut
  {NoStop}%
\bibitem [{\citenamefont {Dobrushin}\ \emph {et~al.}(1985)\citenamefont
  {Dobrushin}, \citenamefont {Kolafa},\ and\ \citenamefont
  {Shlosman}}]{Dobrushin1985}%
  \BibitemOpen
  \bibfield  {author} {\bibinfo {author} {\bibfnamefont {R.~L.}\ \bibnamefont
  {Dobrushin}}, \bibinfo {author} {\bibfnamefont {J.}~\bibnamefont {Kolafa}},\
  and\ \bibinfo {author} {\bibfnamefont {S.~B.}\ \bibnamefont {Shlosman}},\
  }\bibfield  {title} {\bibinfo {title} {Phase diagram of the two-dimensional
  {I}sing antiferromagnet (computer-assisted proof)},\ }\href
  {https://doi.org/10.1007/BF01208821} {\bibfield  {journal} {\bibinfo
  {journal} {Commun. Math. Phys.}\ }\textbf {\bibinfo {volume} {102}},\
  \bibinfo {pages} {89} (\bibinfo {year} {1985})}\BibitemShut {NoStop}%
\bibitem [{\citenamefont {Beauzamy}(1996)}]{Beauzamy1996}%
  \BibitemOpen
  \bibfield  {author} {\bibinfo {author} {\bibfnamefont {B.}~\bibnamefont
  {Beauzamy}},\ }\bibfield  {title} {\bibinfo {title} {On complex {L}ee and
  {Y}ang polynomials},\ }\href {https://doi.org/10.1007/BF02506389} {\bibfield
  {journal} {\bibinfo  {journal} {Commun. Math. Phys.}\ }\textbf {\bibinfo
  {volume} {182}},\ \bibinfo {pages} {177} (\bibinfo {year}
  {1996})}\BibitemShut {NoStop}%
\bibitem [{\citenamefont {Kim}(2004)}]{Kim2004}%
  \BibitemOpen
  \bibfield  {author} {\bibinfo {author} {\bibfnamefont {S.-Y.}\ \bibnamefont
  {Kim}},\ }\bibfield  {title} {\bibinfo {title} {Yang-{L}ee zeros of the
  antiferromagnetic {I}sing model},\ }\href
  {https://doi.org/10.1103/PhysRevLett.93.130604} {\bibfield  {journal}
  {\bibinfo  {journal} {Phys. Rev. Lett.}\ }\textbf {\bibinfo {volume} {93}},\
  \bibinfo {pages} {130604} (\bibinfo {year} {2004})}\BibitemShut {NoStop}%
\bibitem [{\citenamefont {Hwang}\ and\ \citenamefont {Kim}(2010)}]{Hwang2010}%
  \BibitemOpen
  \bibfield  {author} {\bibinfo {author} {\bibfnamefont {C.-O.}\ \bibnamefont
  {Hwang}}\ and\ \bibinfo {author} {\bibfnamefont {S.-Y.}\ \bibnamefont
  {Kim}},\ }\bibfield  {title} {\bibinfo {title} {Yang–{L}ee zeros of
  triangular {I}sing antiferromagnets},\ }\href
  {https://doi.org/10.1016/j.physa.2010.08.050} {\bibfield  {journal} {\bibinfo
   {journal} {Phys. A: Stat. Mech. Appl.}\ }\textbf {\bibinfo {volume} {389}},\
  \bibinfo {pages} {5650} (\bibinfo {year} {2010})}\BibitemShut {NoStop}%
\bibitem [{\citenamefont {Lebowitz}\ \emph {et~al.}(2012)\citenamefont
  {Lebowitz}, \citenamefont {Ruelle},\ and\ \citenamefont
  {Speer}}]{Lebowitz2012}%
  \BibitemOpen
  \bibfield  {author} {\bibinfo {author} {\bibfnamefont {J.~L.}\ \bibnamefont
  {Lebowitz}}, \bibinfo {author} {\bibfnamefont {D.}~\bibnamefont {Ruelle}},\
  and\ \bibinfo {author} {\bibfnamefont {E.~R.}\ \bibnamefont {Speer}},\
  }\bibfield  {title} {\bibinfo {title} {Location of the {L}ee-{Y}ang zeros and
  absence of phase transitions in some {I}sing spin systems},\ }\href
  {https://doi.org/10.1063/1.4738622} {\bibfield  {journal} {\bibinfo
  {journal} {J. Math. Phys.}\ }\textbf {\bibinfo {volume} {53}},\ \bibinfo
  {pages} {095211} (\bibinfo {year} {2012})}\BibitemShut {NoStop}%
\bibitem [{\citenamefont {Lebowitz}\ and\ \citenamefont
  {Scaramazza}(2016)}]{Lebowitz2016}%
  \BibitemOpen
  \bibfield  {author} {\bibinfo {author} {\bibfnamefont {J.~L.}\ \bibnamefont
  {Lebowitz}}\ and\ \bibinfo {author} {\bibfnamefont {J.~A.}\ \bibnamefont
  {Scaramazza}},\ }\bibfield  {title} {\bibinfo {title} {A note on
  {L}ee{\textendash}{Y}ang zeros in the negative half-plane},\ }\href
  {https://doi.org/10.1088/0953-8984/28/41/414004} {\bibfield  {journal}
  {\bibinfo  {journal} {J. Phys.: Condens. Matter}\ }\textbf {\bibinfo {volume}
  {28}},\ \bibinfo {pages} {414004} (\bibinfo {year} {2016})}\BibitemShut
  {NoStop}%
\bibitem [{\citenamefont {Heyl}\ \emph {et~al.}(2013)\citenamefont {Heyl},
  \citenamefont {Polkovnikov},\ and\ \citenamefont {Kehrein}}]{Heyl2013}%
  \BibitemOpen
  \bibfield  {author} {\bibinfo {author} {\bibfnamefont {M.}~\bibnamefont
  {Heyl}}, \bibinfo {author} {\bibfnamefont {A.}~\bibnamefont {Polkovnikov}},\
  and\ \bibinfo {author} {\bibfnamefont {S.}~\bibnamefont {Kehrein}},\
  }\bibfield  {title} {\bibinfo {title} {Dynamical quantum phase transitions in
  the transverse-field {I}sing model},\ }\href
  {https://doi.org/10.1103/PhysRevLett.110.135704} {\bibfield  {journal}
  {\bibinfo  {journal} {Phys. Rev. Lett.}\ }\textbf {\bibinfo {volume} {110}},\
  \bibinfo {pages} {135704} (\bibinfo {year} {2013})}\BibitemShut {NoStop}%
\bibitem [{\citenamefont {Brandner}\ \emph {et~al.}(2017)\citenamefont
  {Brandner}, \citenamefont {Maisi}, \citenamefont {Pekola}, \citenamefont
  {Garrahan},\ and\ \citenamefont {Flindt}}]{Brandner2017}%
  \BibitemOpen
  \bibfield  {author} {\bibinfo {author} {\bibfnamefont {K.}~\bibnamefont
  {Brandner}}, \bibinfo {author} {\bibfnamefont {V.~F.}\ \bibnamefont {Maisi}},
  \bibinfo {author} {\bibfnamefont {J.~P.}\ \bibnamefont {Pekola}}, \bibinfo
  {author} {\bibfnamefont {J.~P.}\ \bibnamefont {Garrahan}},\ and\ \bibinfo
  {author} {\bibfnamefont {C.}~\bibnamefont {Flindt}},\ }\bibfield  {title}
  {\bibinfo {title} {Experimental determination of dynamical {L}ee-{Y}ang
  zeros},\ }\href {https://doi.org/10.1103/PhysRevLett.118.180601} {\bibfield
  {journal} {\bibinfo  {journal} {Phys. Rev. Lett.}\ }\textbf {\bibinfo
  {volume} {118}},\ \bibinfo {pages} {180601} (\bibinfo {year}
  {2017})}\BibitemShut {NoStop}%
\bibitem [{\citenamefont {Deger}\ and\ \citenamefont
  {Flindt}(2019)}]{Deger2019}%
  \BibitemOpen
  \bibfield  {author} {\bibinfo {author} {\bibfnamefont {A.}~\bibnamefont
  {Deger}}\ and\ \bibinfo {author} {\bibfnamefont {C.}~\bibnamefont {Flindt}},\
  }\bibfield  {title} {\bibinfo {title} {Determination of universal critical
  exponents using {L}ee-{Y}ang theory},\ }\href
  {https://doi.org/10.1103/PhysRevResearch.1.023004} {\bibfield  {journal}
  {\bibinfo  {journal} {Phys. Rev. Research}\ }\textbf {\bibinfo {volume}
  {1}},\ \bibinfo {pages} {023004} (\bibinfo {year} {2019})}\BibitemShut
  {NoStop}%
\bibitem [{\citenamefont {Deger}\ \emph {et~al.}(2020)\citenamefont {Deger},
  \citenamefont {Brange},\ and\ \citenamefont {Flindt}}]{Deger2020}%
  \BibitemOpen
  \bibfield  {author} {\bibinfo {author} {\bibfnamefont {A.}~\bibnamefont
  {Deger}}, \bibinfo {author} {\bibfnamefont {F.}~\bibnamefont {Brange}},\ and\
  \bibinfo {author} {\bibfnamefont {C.}~\bibnamefont {Flindt}},\ }\bibfield
  {title} {\bibinfo {title} {Lee-{Y}ang theory, high cumulants, and
  large-deviation statistics of the magnetization in the {I}sing model},\
  }\href {https://doi.org/10.1103/PhysRevB.102.174418} {\bibfield  {journal}
  {\bibinfo  {journal} {Phys. Rev. B}\ }\textbf {\bibinfo {volume} {102}},\
  \bibinfo {pages} {174418} (\bibinfo {year} {2020})}\BibitemShut {NoStop}%
\bibitem [{\citenamefont {Kist}\ \emph {et~al.}(2021)\citenamefont {Kist},
  \citenamefont {Lado},\ and\ \citenamefont {Flindt}}]{Kist2021}%
  \BibitemOpen
  \bibfield  {author} {\bibinfo {author} {\bibfnamefont {T.}~\bibnamefont
  {Kist}}, \bibinfo {author} {\bibfnamefont {J.~L.}\ \bibnamefont {Lado}},\
  and\ \bibinfo {author} {\bibfnamefont {C.}~\bibnamefont {Flindt}},\
  }\bibfield  {title} {\bibinfo {title} {Lee-{Y}ang theory of criticality in
  interacting quantum many-body systems},\ }\href
  {https://doi.org/10.1103/PhysRevResearch.3.033206} {\bibfield  {journal}
  {\bibinfo  {journal} {Phys. Rev. Research}\ }\textbf {\bibinfo {volume}
  {3}},\ \bibinfo {pages} {033206} (\bibinfo {year} {2021})}\BibitemShut
  {NoStop}%
\bibitem [{\citenamefont {Kurtz}(1992)}]{Kurtz1992}%
  \BibitemOpen
  \bibfield  {author} {\bibinfo {author} {\bibfnamefont {D.~C.}\ \bibnamefont
  {Kurtz}},\ }\bibfield  {title} {\bibinfo {title} {A sufficient condition for
  all the roots of a polynomial to be real},\ }\href
  {https://doi.org/10.2307/2325063} {\bibfield  {journal} {\bibinfo  {journal}
  {Am. Math. Mon.}\ }\textbf {\bibinfo {volume} {99}},\ \bibinfo {pages} {259}
  (\bibinfo {year} {1992})}\BibitemShut {NoStop}%
\bibitem [{\citenamefont {Borcea}\ and\ \citenamefont
  {Brändén}(2009{\natexlab{a}})}]{Borcea2009a}%
  \BibitemOpen
  \bibfield  {author} {\bibinfo {author} {\bibfnamefont {J.}~\bibnamefont
  {Borcea}}\ and\ \bibinfo {author} {\bibfnamefont {P.}~\bibnamefont
  {Brändén}},\ }\bibfield  {title} {\bibinfo {title} {The {L}ee-{Y}ang and
  {P}olya-{S}chur programs. {I}. {L}inear operators preserving stability},\
  }\href {https://doi.org/10.1007/s00222-009-0189-3} {\bibfield  {journal}
  {\bibinfo  {journal} {Invent. Math.}\ }\textbf {\bibinfo {volume} {177}},\
  \bibinfo {pages} {541} (\bibinfo {year} {2009}{\natexlab{a}})}\BibitemShut
  {NoStop}%
\bibitem [{\citenamefont {Borcea}\ and\ \citenamefont
  {Brändén}(2009{\natexlab{b}})}]{Borcea2009b}%
  \BibitemOpen
  \bibfield  {author} {\bibinfo {author} {\bibfnamefont {J.}~\bibnamefont
  {Borcea}}\ and\ \bibinfo {author} {\bibfnamefont {P.}~\bibnamefont
  {Brändén}},\ }\bibfield  {title} {\bibinfo {title} {The {L}ee-{Y}ang and
  {P}olya-{S}chur programs. {II}. {T}heory of stable polynomials and
  applications},\ }\href {https://doi.org/10.1002/cpa.20295} {\bibfield
  {journal} {\bibinfo  {journal} {Commun. Pure Appl. Math.}\ }\textbf {\bibinfo
  {volume} {62}},\ \bibinfo {pages} {1595} (\bibinfo {year}
  {2009}{\natexlab{b}})}\BibitemShut {NoStop}%
\bibitem [{\citenamefont {Ruelle}(2010)}]{Ruelle2010}%
  \BibitemOpen
  \bibfield  {author} {\bibinfo {author} {\bibfnamefont {D.}~\bibnamefont
  {Ruelle}},\ }\bibfield  {title} {\bibinfo {title} {Characterization of
  {Lee}-{Yang} polynomials},\ }\href
  {https://doi.org/10.4007/annals.2010.171.589} {\bibfield  {journal} {\bibinfo
   {journal} {Ann. Math.}\ }\textbf {\bibinfo {volume} {171}},\ \bibinfo
  {pages} {589} (\bibinfo {year} {2010})}\BibitemShut {NoStop}%
\bibitem [{\citenamefont {Wei}\ and\ \citenamefont {Liu}(2012)}]{Wei2012}%
  \BibitemOpen
  \bibfield  {author} {\bibinfo {author} {\bibfnamefont {B.-B.}\ \bibnamefont
  {Wei}}\ and\ \bibinfo {author} {\bibfnamefont {R.-B.}\ \bibnamefont {Liu}},\
  }\bibfield  {title} {\bibinfo {title} {Lee-{Y}ang zeros and critical times in
  decoherence of a probe spin coupled to a bath},\ }\href
  {https://doi.org/10.1103/PhysRevLett.109.185701} {\bibfield  {journal}
  {\bibinfo  {journal} {Phys. Rev. Lett.}\ }\textbf {\bibinfo {volume} {109}},\
  \bibinfo {pages} {185701} (\bibinfo {year} {2012})}\BibitemShut {NoStop}%
\bibitem [{\citenamefont {Peng}\ \emph {et~al.}(2015)\citenamefont {Peng},
  \citenamefont {Zhou}, \citenamefont {Wei}, \citenamefont {Cui}, \citenamefont
  {Du},\ and\ \citenamefont {Liu}}]{Peng2015}%
  \BibitemOpen
  \bibfield  {author} {\bibinfo {author} {\bibfnamefont {X.}~\bibnamefont
  {Peng}}, \bibinfo {author} {\bibfnamefont {H.}~\bibnamefont {Zhou}}, \bibinfo
  {author} {\bibfnamefont {B.-B.}\ \bibnamefont {Wei}}, \bibinfo {author}
  {\bibfnamefont {J.}~\bibnamefont {Cui}}, \bibinfo {author} {\bibfnamefont
  {J.}~\bibnamefont {Du}},\ and\ \bibinfo {author} {\bibfnamefont {R.-B.}\
  \bibnamefont {Liu}},\ }\bibfield  {title} {\bibinfo {title} {Experimental
  observation of {L}ee-{Y}ang zeros},\ }\href
  {https://doi.org/10.1103/PhysRevLett.114.010601} {\bibfield  {journal}
  {\bibinfo  {journal} {Phys. Rev. Lett.}\ }\textbf {\bibinfo {volume} {114}},\
  \bibinfo {pages} {010601} (\bibinfo {year} {2015})}\BibitemShut {NoStop}%
\bibitem [{\citenamefont {Bernien}\ \emph {et~al.}(2017)\citenamefont
  {Bernien}, \citenamefont {Schwartz}, \citenamefont {Keesling}, \citenamefont
  {Levine}, \citenamefont {Omran}, \citenamefont {Pichler}, \citenamefont
  {Choi}, \citenamefont {Zibrov}, \citenamefont {Endres}, \citenamefont
  {Greiner}, \citenamefont {Vuleti{\'{c}}},\ and\ \citenamefont
  {Lukin}}]{Bernien2017}%
  \BibitemOpen
  \bibfield  {author} {\bibinfo {author} {\bibfnamefont {H.}~\bibnamefont
  {Bernien}}, \bibinfo {author} {\bibfnamefont {S.}~\bibnamefont {Schwartz}},
  \bibinfo {author} {\bibfnamefont {A.}~\bibnamefont {Keesling}}, \bibinfo
  {author} {\bibfnamefont {H.}~\bibnamefont {Levine}}, \bibinfo {author}
  {\bibfnamefont {A.}~\bibnamefont {Omran}}, \bibinfo {author} {\bibfnamefont
  {H.}~\bibnamefont {Pichler}}, \bibinfo {author} {\bibfnamefont
  {S.}~\bibnamefont {Choi}}, \bibinfo {author} {\bibfnamefont {A.~S.}\
  \bibnamefont {Zibrov}}, \bibinfo {author} {\bibfnamefont {M.}~\bibnamefont
  {Endres}}, \bibinfo {author} {\bibfnamefont {M.}~\bibnamefont {Greiner}},
  \bibinfo {author} {\bibfnamefont {V.}~\bibnamefont {Vuleti{\'{c}}}},\ and\
  \bibinfo {author} {\bibfnamefont {M.~D.}\ \bibnamefont {Lukin}},\ }\bibfield
  {title} {\bibinfo {title} {Probing many-body dynamics on a 51-atom quantum
  simulator},\ }\href {https://doi.org/10.1038/nature24622} {\bibfield
  {journal} {\bibinfo  {journal} {Nature}\ }\textbf {\bibinfo {volume} {551}},\
  \bibinfo {pages} {579} (\bibinfo {year} {2017})}\BibitemShut {NoStop}%
\bibitem [{\citenamefont {Keesling}\ \emph {et~al.}(2019)\citenamefont
  {Keesling}, \citenamefont {Omran}, \citenamefont {Levine}, \citenamefont
  {Bernien}, \citenamefont {Pichler}, \citenamefont {Choi}, \citenamefont
  {Samajdar}, \citenamefont {Schwartz}, \citenamefont {Silvi}, \citenamefont
  {Sachdev}, \citenamefont {Zoller}, \citenamefont {Endres}, \citenamefont
  {Greiner}, \citenamefont {Vuleti{\'{c}}},\ and\ \citenamefont
  {Lukin}}]{Keesling2019}%
  \BibitemOpen
  \bibfield  {author} {\bibinfo {author} {\bibfnamefont {A.}~\bibnamefont
  {Keesling}}, \bibinfo {author} {\bibfnamefont {A.}~\bibnamefont {Omran}},
  \bibinfo {author} {\bibfnamefont {H.}~\bibnamefont {Levine}}, \bibinfo
  {author} {\bibfnamefont {H.}~\bibnamefont {Bernien}}, \bibinfo {author}
  {\bibfnamefont {H.}~\bibnamefont {Pichler}}, \bibinfo {author} {\bibfnamefont
  {S.}~\bibnamefont {Choi}}, \bibinfo {author} {\bibfnamefont {R.}~\bibnamefont
  {Samajdar}}, \bibinfo {author} {\bibfnamefont {S.}~\bibnamefont {Schwartz}},
  \bibinfo {author} {\bibfnamefont {P.}~\bibnamefont {Silvi}}, \bibinfo
  {author} {\bibfnamefont {S.}~\bibnamefont {Sachdev}}, \bibinfo {author}
  {\bibfnamefont {P.}~\bibnamefont {Zoller}}, \bibinfo {author} {\bibfnamefont
  {M.}~\bibnamefont {Endres}}, \bibinfo {author} {\bibfnamefont
  {M.}~\bibnamefont {Greiner}}, \bibinfo {author} {\bibfnamefont
  {V.}~\bibnamefont {Vuleti{\'{c}}}},\ and\ \bibinfo {author} {\bibfnamefont
  {M.~D.}\ \bibnamefont {Lukin}},\ }\bibfield  {title} {\bibinfo {title}
  {Quantum {K}ibble--{Z}urek mechanism and critical dynamics on a programmable
  {R}ydberg simulator},\ }\href {https://doi.org/10.1038/s41586-019-1070-1}
  {\bibfield  {journal} {\bibinfo  {journal} {Nature}\ }\textbf {\bibinfo
  {volume} {568}},\ \bibinfo {pages} {207} (\bibinfo {year}
  {2019})}\BibitemShut {NoStop}%
\bibitem [{\citenamefont {Satzinger}\ \emph {et~al.}(2021)\citenamefont
  {Satzinger} \emph {et~al.}}]{Satzinger2021}%
  \BibitemOpen
  \bibfield  {author} {\bibinfo {author} {\bibfnamefont {K.~J.}\ \bibnamefont
  {Satzinger}} \emph {et~al.},\ }\bibfield  {title} {\bibinfo {title}
  {Realizing topologically ordered states on a quantum processor},\ }\href
  {https://doi.org/10.1126/science.abi8378} {\bibfield  {journal} {\bibinfo
  {journal} {Science}\ }\textbf {\bibinfo {volume} {374}},\ \bibinfo {pages}
  {1237} (\bibinfo {year} {2021})}\BibitemShut {NoStop}%
\bibitem [{\citenamefont {Ebadi}\ \emph {et~al.}(2021)\citenamefont {Ebadi},
  \citenamefont {Wang}, \citenamefont {Levine}, \citenamefont {Keesling},
  \citenamefont {Semeghini}, \citenamefont {Omran}, \citenamefont {Bluvstein},
  \citenamefont {Samajdar}, \citenamefont {Pichler}, \citenamefont {Ho},
  \citenamefont {Choi}, \citenamefont {Sachdev}, \citenamefont {Greiner},
  \citenamefont {Vuleti{\'{c}}},\ and\ \citenamefont {Lukin}}]{Ebadi2021}%
  \BibitemOpen
  \bibfield  {author} {\bibinfo {author} {\bibfnamefont {S.}~\bibnamefont
  {Ebadi}}, \bibinfo {author} {\bibfnamefont {T.~T.}\ \bibnamefont {Wang}},
  \bibinfo {author} {\bibfnamefont {H.}~\bibnamefont {Levine}}, \bibinfo
  {author} {\bibfnamefont {A.}~\bibnamefont {Keesling}}, \bibinfo {author}
  {\bibfnamefont {G.}~\bibnamefont {Semeghini}}, \bibinfo {author}
  {\bibfnamefont {A.}~\bibnamefont {Omran}}, \bibinfo {author} {\bibfnamefont
  {D.}~\bibnamefont {Bluvstein}}, \bibinfo {author} {\bibfnamefont
  {R.}~\bibnamefont {Samajdar}}, \bibinfo {author} {\bibfnamefont
  {H.}~\bibnamefont {Pichler}}, \bibinfo {author} {\bibfnamefont {W.~W.}\
  \bibnamefont {Ho}}, \bibinfo {author} {\bibfnamefont {S.}~\bibnamefont
  {Choi}}, \bibinfo {author} {\bibfnamefont {S.}~\bibnamefont {Sachdev}},
  \bibinfo {author} {\bibfnamefont {M.}~\bibnamefont {Greiner}}, \bibinfo
  {author} {\bibfnamefont {V.}~\bibnamefont {Vuleti{\'{c}}}},\ and\ \bibinfo
  {author} {\bibfnamefont {M.~D.}\ \bibnamefont {Lukin}},\ }\bibfield  {title}
  {\bibinfo {title} {Quantum phases of matter on a 256-atom programmable
  quantum simulator},\ }\href {https://doi.org/10.1038/s41586-021-03582-4}
  {\bibfield  {journal} {\bibinfo  {journal} {Nature}\ }\textbf {\bibinfo
  {volume} {595}},\ \bibinfo {pages} {227} (\bibinfo {year}
  {2021})}\BibitemShut {NoStop}%
\bibitem [{\citenamefont {Samajdar}\ \emph {et~al.}(2020)\citenamefont
  {Samajdar}, \citenamefont {Ho}, \citenamefont {Pichler}, \citenamefont
  {Lukin},\ and\ \citenamefont {Sachdev}}]{Samajdar2020}%
  \BibitemOpen
  \bibfield  {author} {\bibinfo {author} {\bibfnamefont {R.}~\bibnamefont
  {Samajdar}}, \bibinfo {author} {\bibfnamefont {W.~W.}\ \bibnamefont {Ho}},
  \bibinfo {author} {\bibfnamefont {H.}~\bibnamefont {Pichler}}, \bibinfo
  {author} {\bibfnamefont {M.~D.}\ \bibnamefont {Lukin}},\ and\ \bibinfo
  {author} {\bibfnamefont {S.}~\bibnamefont {Sachdev}},\ }\bibfield  {title}
  {\bibinfo {title} {Complex density wave orders and quantum phase transitions
  in a model of square-lattice {R}ydberg atom arrays},\ }\href
  {https://doi.org/10.1103/PhysRevLett.124.103601} {\bibfield  {journal}
  {\bibinfo  {journal} {Phys. Rev. Lett.}\ }\textbf {\bibinfo {volume} {124}},\
  \bibinfo {pages} {103601} (\bibinfo {year} {2020})}\BibitemShut {NoStop}%
\bibitem [{\citenamefont {Kalinowski}\ \emph {et~al.}(2022)\citenamefont
  {Kalinowski}, \citenamefont {Samajdar}, \citenamefont {Melko}, \citenamefont
  {Lukin}, \citenamefont {Sachdev},\ and\ \citenamefont
  {Choi}}]{Kalinowski2022}%
  \BibitemOpen
  \bibfield  {author} {\bibinfo {author} {\bibfnamefont {M.}~\bibnamefont
  {Kalinowski}}, \bibinfo {author} {\bibfnamefont {R.}~\bibnamefont
  {Samajdar}}, \bibinfo {author} {\bibfnamefont {R.~G.}\ \bibnamefont {Melko}},
  \bibinfo {author} {\bibfnamefont {M.~D.}\ \bibnamefont {Lukin}}, \bibinfo
  {author} {\bibfnamefont {S.}~\bibnamefont {Sachdev}},\ and\ \bibinfo {author}
  {\bibfnamefont {S.}~\bibnamefont {Choi}},\ }\bibfield  {title} {\bibinfo
  {title} {Bulk and boundary quantum phase transitions in a square {R}ydberg
  atom array},\ }\href {https://doi.org/10.1103/PhysRevB.105.174417} {\bibfield
   {journal} {\bibinfo  {journal} {Phys. Rev. B}\ }\textbf {\bibinfo {volume}
  {105}},\ \bibinfo {pages} {174417} (\bibinfo {year} {2022})}\BibitemShut
  {NoStop}%
\bibitem [{\citenamefont {Verresen}\ \emph {et~al.}(2021)\citenamefont
  {Verresen}, \citenamefont {Lukin},\ and\ \citenamefont
  {Vishwanath}}]{Verresen2021}%
  \BibitemOpen
  \bibfield  {author} {\bibinfo {author} {\bibfnamefont {R.}~\bibnamefont
  {Verresen}}, \bibinfo {author} {\bibfnamefont {M.~D.}\ \bibnamefont
  {Lukin}},\ and\ \bibinfo {author} {\bibfnamefont {A.}~\bibnamefont
  {Vishwanath}},\ }\bibfield  {title} {\bibinfo {title} {Prediction of toric
  code topological order from {R}ydberg blockade},\ }\href
  {https://doi.org/10.1103/PhysRevX.11.031005} {\bibfield  {journal} {\bibinfo
  {journal} {Phys. Rev. X}\ }\textbf {\bibinfo {volume} {11}},\ \bibinfo
  {pages} {031005} (\bibinfo {year} {2021})}\BibitemShut {NoStop}%
\bibitem [{\citenamefont {Semeghini}\ \emph {et~al.}(2021)\citenamefont
  {Semeghini}, \citenamefont {Levine}, \citenamefont {Keesling}, \citenamefont
  {Ebadi}, \citenamefont {Wang}, \citenamefont {Bluvstein}, \citenamefont
  {Verresen}, \citenamefont {Pichler}, \citenamefont {Kalinowski},
  \citenamefont {Samajdar}, \citenamefont {Omran}, \citenamefont {Sachdev},
  \citenamefont {Vishwanath}, \citenamefont {Greiner}, \citenamefont
  {Vuletić},\ and\ \citenamefont {Lukin}}]{Semeghini2021}%
  \BibitemOpen
  \bibfield  {author} {\bibinfo {author} {\bibfnamefont {G.}~\bibnamefont
  {Semeghini}}, \bibinfo {author} {\bibfnamefont {H.}~\bibnamefont {Levine}},
  \bibinfo {author} {\bibfnamefont {A.}~\bibnamefont {Keesling}}, \bibinfo
  {author} {\bibfnamefont {S.}~\bibnamefont {Ebadi}}, \bibinfo {author}
  {\bibfnamefont {T.~T.}\ \bibnamefont {Wang}}, \bibinfo {author}
  {\bibfnamefont {D.}~\bibnamefont {Bluvstein}}, \bibinfo {author}
  {\bibfnamefont {R.}~\bibnamefont {Verresen}}, \bibinfo {author}
  {\bibfnamefont {H.}~\bibnamefont {Pichler}}, \bibinfo {author} {\bibfnamefont
  {M.}~\bibnamefont {Kalinowski}}, \bibinfo {author} {\bibfnamefont
  {R.}~\bibnamefont {Samajdar}}, \bibinfo {author} {\bibfnamefont
  {A.}~\bibnamefont {Omran}}, \bibinfo {author} {\bibfnamefont
  {S.}~\bibnamefont {Sachdev}}, \bibinfo {author} {\bibfnamefont
  {A.}~\bibnamefont {Vishwanath}}, \bibinfo {author} {\bibfnamefont
  {M.}~\bibnamefont {Greiner}}, \bibinfo {author} {\bibfnamefont
  {V.}~\bibnamefont {Vuletić}},\ and\ \bibinfo {author} {\bibfnamefont
  {M.~D.}\ \bibnamefont {Lukin}},\ }\bibfield  {title} {\bibinfo {title}
  {Probing topological spin liquids on a programmable quantum simulator},\
  }\href {https://doi.org/10.1126/science.abi8794} {\bibfield  {journal}
  {\bibinfo  {journal} {Science}\ }\textbf {\bibinfo {volume} {374}},\ \bibinfo
  {pages} {1242} (\bibinfo {year} {2021})}\BibitemShut {NoStop}%
\bibitem [{\citenamefont {Samajdar}\ \emph {et~al.}(2021)\citenamefont
  {Samajdar}, \citenamefont {Ho}, \citenamefont {Pichler}, \citenamefont
  {Lukin},\ and\ \citenamefont {Sachdev}}]{Samajdar2021}%
  \BibitemOpen
  \bibfield  {author} {\bibinfo {author} {\bibfnamefont {R.}~\bibnamefont
  {Samajdar}}, \bibinfo {author} {\bibfnamefont {W.~W.}\ \bibnamefont {Ho}},
  \bibinfo {author} {\bibfnamefont {H.}~\bibnamefont {Pichler}}, \bibinfo
  {author} {\bibfnamefont {M.~D.}\ \bibnamefont {Lukin}},\ and\ \bibinfo
  {author} {\bibfnamefont {S.}~\bibnamefont {Sachdev}},\ }\bibfield  {title}
  {\bibinfo {title} {Quantum phases of {R}ydberg atoms on a kagome lattice},\
  }\href {https://doi.org/10.1073/pnas.2015785118} {\bibfield  {journal}
  {\bibinfo  {journal} {Proc. Natl. Acad. Sci. U.S.A.}\ }\textbf {\bibinfo
  {volume} {118}},\ \bibinfo {pages} {e2015785118} (\bibinfo {year}
  {2021})}\BibitemShut {NoStop}%
\bibitem [{\citenamefont {Cheng}\ \emph {et~al.}(2021)\citenamefont {Cheng},
  \citenamefont {Li},\ and\ \citenamefont {Zhai}}]{Cheng2021}%
  \BibitemOpen
  \bibfield  {author} {\bibinfo {author} {\bibfnamefont {Y.}~\bibnamefont
  {Cheng}}, \bibinfo {author} {\bibfnamefont {C.}~\bibnamefont {Li}},\ and\
  \bibinfo {author} {\bibfnamefont {H.}~\bibnamefont {Zhai}},\ }\bibfield
  {title} {\bibinfo {title} {Variational approach to quantum spin liquid in a
  {R}ydberg atom simulator},\ }\href@noop {} {\bibfield  {journal} {\bibinfo
  {journal} {arXiv:2112.13688}\ } (\bibinfo {year} {2021})}\BibitemShut
  {NoStop}%
\bibitem [{\citenamefont {Giudici}\ \emph {et~al.}(2022)\citenamefont
  {Giudici}, \citenamefont {Lukin},\ and\ \citenamefont
  {Pichler}}]{Giudici2022}%
  \BibitemOpen
  \bibfield  {author} {\bibinfo {author} {\bibfnamefont {G.}~\bibnamefont
  {Giudici}}, \bibinfo {author} {\bibfnamefont {M.~D.}\ \bibnamefont {Lukin}},\
  and\ \bibinfo {author} {\bibfnamefont {H.}~\bibnamefont {Pichler}},\
  }\bibfield  {title} {\bibinfo {title} {Dynamical preparation of quantum spin
  liquids in {R}ydberg atom arrays},\ }\href
  {https://doi.org/10.1103/PhysRevLett.129.090401} {\bibfield  {journal}
  {\bibinfo  {journal} {Phys. Rev. Lett.}\ }\textbf {\bibinfo {volume} {129}},\
  \bibinfo {pages} {090401} (\bibinfo {year} {2022})}\BibitemShut {NoStop}%
\bibitem [{\citenamefont {Fendley}\ \emph {et~al.}(2004)\citenamefont
  {Fendley}, \citenamefont {Sengupta},\ and\ \citenamefont
  {Sachdev}}]{Fendley2004}%
  \BibitemOpen
  \bibfield  {author} {\bibinfo {author} {\bibfnamefont {P.}~\bibnamefont
  {Fendley}}, \bibinfo {author} {\bibfnamefont {K.}~\bibnamefont {Sengupta}},\
  and\ \bibinfo {author} {\bibfnamefont {S.}~\bibnamefont {Sachdev}},\
  }\bibfield  {title} {\bibinfo {title} {Competing density-wave orders in a
  one-dimensional hard-boson model},\ }\href
  {https://doi.org/10.1103/PhysRevB.69.075106} {\bibfield  {journal} {\bibinfo
  {journal} {Phys. Rev. B}\ }\textbf {\bibinfo {volume} {69}},\ \bibinfo
  {pages} {075106} (\bibinfo {year} {2004})}\BibitemShut {NoStop}%
\bibitem [{\citenamefont {Samajdar}\ \emph {et~al.}(2018)\citenamefont
  {Samajdar}, \citenamefont {Choi}, \citenamefont {Pichler}, \citenamefont
  {Lukin},\ and\ \citenamefont {Sachdev}}]{Samajdar2018}%
  \BibitemOpen
  \bibfield  {author} {\bibinfo {author} {\bibfnamefont {R.}~\bibnamefont
  {Samajdar}}, \bibinfo {author} {\bibfnamefont {S.}~\bibnamefont {Choi}},
  \bibinfo {author} {\bibfnamefont {H.}~\bibnamefont {Pichler}}, \bibinfo
  {author} {\bibfnamefont {M.~D.}\ \bibnamefont {Lukin}},\ and\ \bibinfo
  {author} {\bibfnamefont {S.}~\bibnamefont {Sachdev}},\ }\bibfield  {title}
  {\bibinfo {title} {Numerical study of the chiral {${\mathbb{Z}}_{3}$} quantum
  phase transition in one spatial dimension},\ }\href
  {https://doi.org/10.1103/PhysRevA.98.023614} {\bibfield  {journal} {\bibinfo
  {journal} {Phys. Rev. A}\ }\textbf {\bibinfo {volume} {98}},\ \bibinfo
  {pages} {023614} (\bibinfo {year} {2018})}\BibitemShut {NoStop}%
\bibitem [{\citenamefont {Giudici}\ \emph {et~al.}(2019)\citenamefont
  {Giudici}, \citenamefont {Angelone}, \citenamefont {Magnifico}, \citenamefont
  {Zeng}, \citenamefont {Giudice}, \citenamefont {Mendes-Santos},\ and\
  \citenamefont {Dalmonte}}]{Giudici2019}%
  \BibitemOpen
  \bibfield  {author} {\bibinfo {author} {\bibfnamefont {G.}~\bibnamefont
  {Giudici}}, \bibinfo {author} {\bibfnamefont {A.}~\bibnamefont {Angelone}},
  \bibinfo {author} {\bibfnamefont {G.}~\bibnamefont {Magnifico}}, \bibinfo
  {author} {\bibfnamefont {Z.}~\bibnamefont {Zeng}}, \bibinfo {author}
  {\bibfnamefont {G.}~\bibnamefont {Giudice}}, \bibinfo {author} {\bibfnamefont
  {T.}~\bibnamefont {Mendes-Santos}},\ and\ \bibinfo {author} {\bibfnamefont
  {M.}~\bibnamefont {Dalmonte}},\ }\bibfield  {title} {\bibinfo {title}
  {Diagnosing {P}otts criticality and two-stage melting in one-dimensional
  hard-core boson models},\ }\href {https://doi.org/10.1103/PhysRevB.99.094434}
  {\bibfield  {journal} {\bibinfo  {journal} {Phys. Rev. B}\ }\textbf {\bibinfo
  {volume} {99}},\ \bibinfo {pages} {094434} (\bibinfo {year}
  {2019})}\BibitemShut {NoStop}%
\bibitem [{\citenamefont {Chepiga}\ and\ \citenamefont
  {Mila}(2019)}]{Chepiga2019}%
  \BibitemOpen
  \bibfield  {author} {\bibinfo {author} {\bibfnamefont {N.}~\bibnamefont
  {Chepiga}}\ and\ \bibinfo {author} {\bibfnamefont {F.}~\bibnamefont {Mila}},\
  }\bibfield  {title} {\bibinfo {title} {Floating phase versus chiral
  transition in a 1d hard-boson model},\ }\href
  {https://doi.org/10.1103/PhysRevLett.122.017205} {\bibfield  {journal}
  {\bibinfo  {journal} {Phys. Rev. Lett.}\ }\textbf {\bibinfo {volume} {122}},\
  \bibinfo {pages} {017205} (\bibinfo {year} {2019})}\BibitemShut {NoStop}%
\bibitem [{\citenamefont {Rader}\ and\ \citenamefont
  {L\"auchli}(2019)}]{Rader2019}%
  \BibitemOpen
  \bibfield  {author} {\bibinfo {author} {\bibfnamefont {M.}~\bibnamefont
  {Rader}}\ and\ \bibinfo {author} {\bibfnamefont {A.~M.}\ \bibnamefont
  {L\"auchli}},\ }\bibfield  {title} {\bibinfo {title} {Floating phases in
  one-dimensional {R}ydberg {I}sing chains},\ }\href@noop {} {\bibfield
  {journal} {\bibinfo  {journal} {arXiv:1908.02068}\ } (\bibinfo {year}
  {2019})}\BibitemShut {NoStop}%
\bibitem [{\citenamefont {Maceira}\ \emph {et~al.}(2022)\citenamefont
  {Maceira}, \citenamefont {Chepiga},\ and\ \citenamefont
  {Mila}}]{Maceira2022}%
  \BibitemOpen
  \bibfield  {author} {\bibinfo {author} {\bibfnamefont {I.~A.}\ \bibnamefont
  {Maceira}}, \bibinfo {author} {\bibfnamefont {N.}~\bibnamefont {Chepiga}},\
  and\ \bibinfo {author} {\bibfnamefont {F.}~\bibnamefont {Mila}},\ }\bibfield
  {title} {\bibinfo {title} {Conformal and chiral phase transitions in
  {R}ydberg chains},\ }\href@noop {} {\bibfield  {journal} {\bibinfo  {journal}
  {arXiv:2203.01163}\ } (\bibinfo {year} {2022})}\BibitemShut {NoStop}%
\bibitem [{\citenamefont {Turner}\ \emph {et~al.}(2018)\citenamefont {Turner},
  \citenamefont {Michailidis}, \citenamefont {Abanin}, \citenamefont {Serbyn},\
  and\ \citenamefont {Papi{\'{c}}}}]{Turner2018}%
  \BibitemOpen
  \bibfield  {author} {\bibinfo {author} {\bibfnamefont {C.~J.}\ \bibnamefont
  {Turner}}, \bibinfo {author} {\bibfnamefont {A.~A.}\ \bibnamefont
  {Michailidis}}, \bibinfo {author} {\bibfnamefont {D.~A.}\ \bibnamefont
  {Abanin}}, \bibinfo {author} {\bibfnamefont {M.}~\bibnamefont {Serbyn}},\
  and\ \bibinfo {author} {\bibfnamefont {Z.}~\bibnamefont {Papi{\'{c}}}},\
  }\bibfield  {title} {\bibinfo {title} {Weak ergodicity breaking from quantum
  many-body scars},\ }\href {https://doi.org/10.1038/s41567-018-0137-5}
  {\bibfield  {journal} {\bibinfo  {journal} {Nat. Phys.}\ }\textbf {\bibinfo
  {volume} {14}},\ \bibinfo {pages} {745} (\bibinfo {year} {2018})}\BibitemShut
  {NoStop}%
\bibitem [{\citenamefont {Serbyn}\ \emph {et~al.}(2021)\citenamefont {Serbyn},
  \citenamefont {Abanin},\ and\ \citenamefont {Papi{\'{c}}}}]{Serbyn2021}%
  \BibitemOpen
  \bibfield  {author} {\bibinfo {author} {\bibfnamefont {M.}~\bibnamefont
  {Serbyn}}, \bibinfo {author} {\bibfnamefont {D.~A.}\ \bibnamefont {Abanin}},\
  and\ \bibinfo {author} {\bibfnamefont {Z.}~\bibnamefont {Papi{\'{c}}}},\
  }\bibfield  {title} {\bibinfo {title} {Quantum many-body scars and weak
  breaking of ergodicity},\ }\href {https://doi.org/10.1038/s41567-021-01230-2}
  {\bibfield  {journal} {\bibinfo  {journal} {Nat. Phys.}\ }\textbf {\bibinfo
  {volume} {17}},\ \bibinfo {pages} {675} (\bibinfo {year} {2021})}\BibitemShut
  {NoStop}%
\bibitem [{\citenamefont {Alcaraz}\ and\ \citenamefont
  {Pimenta}(2020{\natexlab{a}})}]{Alcaraz2020a}%
  \BibitemOpen
  \bibfield  {author} {\bibinfo {author} {\bibfnamefont {F.~C.}\ \bibnamefont
  {Alcaraz}}\ and\ \bibinfo {author} {\bibfnamefont {R.~A.}\ \bibnamefont
  {Pimenta}},\ }\bibfield  {title} {\bibinfo {title} {Free fermionic and
  parafermionic quantum spin chains with multispin interactions},\ }\href
  {https://doi.org/10.1103/PhysRevB.102.121101} {\bibfield  {journal} {\bibinfo
   {journal} {Phys. Rev. B}\ }\textbf {\bibinfo {volume} {102}},\ \bibinfo
  {pages} {121101(R)} (\bibinfo {year} {2020}{\natexlab{a}})}\BibitemShut
  {NoStop}%
\bibitem [{\citenamefont {Alcaraz}\ and\ \citenamefont
  {Pimenta}(2020{\natexlab{b}})}]{Alcaraz2020b}%
  \BibitemOpen
  \bibfield  {author} {\bibinfo {author} {\bibfnamefont {F.~C.}\ \bibnamefont
  {Alcaraz}}\ and\ \bibinfo {author} {\bibfnamefont {R.~A.}\ \bibnamefont
  {Pimenta}},\ }\bibfield  {title} {\bibinfo {title} {Integrable quantum spin
  chains with free fermionic and parafermionic spectrum},\ }\href
  {https://doi.org/10.1103/PhysRevB.102.235170} {\bibfield  {journal} {\bibinfo
   {journal} {Phys. Rev. B}\ }\textbf {\bibinfo {volume} {102}},\ \bibinfo
  {pages} {235170} (\bibinfo {year} {2020}{\natexlab{b}})}\BibitemShut
  {NoStop}%
\bibitem [{\citenamefont {Fendley}(2019)}]{Fendley2019}%
  \BibitemOpen
  \bibfield  {author} {\bibinfo {author} {\bibfnamefont {P.}~\bibnamefont
  {Fendley}},\ }\bibfield  {title} {\bibinfo {title} {Free fermions in
  disguise},\ }\href {https://doi.org/10.1088/1751-8121/ab305d} {\bibfield
  {journal} {\bibinfo  {journal} {J. Phys. A: Math. Theor.}\ }\textbf {\bibinfo
  {volume} {52}},\ \bibinfo {pages} {335002} (\bibinfo {year}
  {2019})}\BibitemShut {NoStop}%
\end{thebibliography}%
\end{document}